\documentclass[11pt]{article}
\usepackage{amsmath,amssymb,amsfonts}

\usepackage{fullpage}
\usepackage{amsmath}
\usepackage{times}

\usepackage{float}

\usepackage{amsmath}

\usepackage{float}

\newcommand{\comment}[1]{}

\floatstyle{ruled}
\newfloat{algorithm}{h}{loa}
\floatname{algorithm}{Algorithm}

\newtheorem{definition}{Definition}
\newtheorem{theorem}{Theorem}
\newtheorem{corollary}{Corollary}
\newtheorem{lemma}{Lemma}
\newcommand{\qed}{\mbox{\ \ \ }\rule{6pt}{7pt}}%
\newenvironment{proof}{\par{\bf Proof:}}{\qed \par}
\newenvironment{sketch}{\par{\sc Proof Sketch:}} {\qed \par}

\newcommand{\poly}{\mathrm{poly}}

\newcommand{\al}{\alpha}
\newcommand{\be}{\beta}
\newcommand{\ga}{\gamma}
\newcommand{\ta}{\theta}
\newcommand{\de}{\delta}
\newcommand{\ve}{\varepsilon}

\newcommand{\listt}{{\mathcal L}}
\newcommand{\s}{s}
\newcommand{\ts}{\tilde{s}}
\newcommand{\tS}{\tilde{S}}
\newcommand{\aff}{a}

\newcommand{\vn}{v}
\newtheorem{claim}{Claim}
\renewcommand{\Pr}{{\bf Pr}}

\def\sizeof#1{\left|#1  \right|}

\def\setof#1{\left\{{\let\st\colon #1 }\right\}}

\newcommand{\W}{\mathrm{A}}
\newcommand{\w}{\mathrm{a}}

\newcommand{\fracconnect}{\theta}

\begin{document}

% Page heads
%\markboth{M.F. Balcan, C. Borgs, M. Braverman, J. Chayes, and S.H. Teng}{I like her more than you: self-determined communities}
%
%% Title portion
%\title{Finding Endogenously Formed Communities}
%\author{Maria-Florina Balcan
%\affil{Georgia Institute of Technology}
% Christian Borgs \affil{Microsoft Research, New England}
%Mark Braverman
%\affil{Princeton University and University of Toronto}
% Jennifer Chayes \affil{Microsoft Research, New England}
% Shang-Hua Teng \affil{University of Southern California}
%}

\title{Finding Endogenously Formed Communities}

\author{
Maria-Florina Balcan\footnote{School of Computer Science, College of Computing, Georgia Institute of Technology, Atlanta, Georgia.}
\and Christian Borgs\footnote{Microsoft Research, New England, Cambridge, MA.}
\and Mark Braverman \footnote{Princeton University and University of Toronto}
\and Jennifer Chayes\footnote{Microsoft Research, New England,  Cambridge, MA.}
\and Shang-Hua Teng\footnote{Computer Science Department, University of Southern California.}
}

\maketitle

\begin{abstract}
A central problem in e-commerce is determining overlapping communities
(clusters) among
individuals or objects in the absence of external identification or tagging.
We address this problem by introducing a framework that captures the notion
of communities or clusters determined by the relative affinities among their members.
To this end we define what we call an affinity system, which is a set of elements,
  each with a vector characterizing its preference for all other elements in the set.
%The preference of a member can be given either by a ranking of all members or
%  by a weighted vector that defines the degrees of its affinity to others.
%Affinity systems are useful for modeling social systems as well as general data sets,
% as social interactions are often determined by affinities among the members.
We define a natural notion of (potentially overlapping) communities in an affinity system,
  in which the members of a given community collectively prefer each other to anyone else
  outside the community.
Thus these communities are endogenously formed in the affinity system
and are ``self-determined'' or ``self-certified'' by its members.

We provide a tight polynomial bound on the number of self-determined communities
     as a function of the robustness of the community.
We present a polynomial-time algorithm for enumerating these communities. Moreover,
we obtain a local algorithm with a strong stochastic performance guarantee
    that can find a community in time nearly linear in the of size the community (as opposed to the size of the network).

Social networks and social interactions fit particularly naturally
within the affinity system framework -- if we can appropriately
extract the affinities from the relatively sparse yet rich
information from social networks and social interactions,
 our analysis
then yields a set of efficient algorithms for enumerating self-determined communities in social networks. In the context of social networks we also connect our analysis with results about $(\alpha,\beta)$-clusters introduced by Mishra, Schreiber, Stanton, and Tarjan~\cite{mishra-conf,mishra}.
In contrast with the polynomial bound we prove on the number of communities in the affinity system model, we show that there exists a family of networks with superpolynomial
    number of $(\alpha,\beta)$-clusters.
%Our analysis also sheds light on the $(\alpha,\beta)$-clusters introduced by~\cite{mishra-conf,mishra}
%     for analyzing social networks.
%We prove that there exists a family of networks with superpolynomial
%    number of $(\alpha,\beta)$-clusters.
%We also show that under the assumption that the planted clique problem is hard, even finding an $(\alpha,\beta)$-cluster
 % is computational hard.
%In contrast, we show that self-determined communities of a social network can be enumerated in polynomial time.

%Since social networks and social interactions fit particularly naturally within the affinity system framework, our analysis also sheds light on the difficult/timley problem  of analyzing communities social networks.

\end{abstract}

%\category{F.2.2}{Nonnumerical Algorithms and Problems}{}
%%\category{F.2.0}{Analysis of Algorithms and Problem Complexity}{General}
%\category{J.4 }{Social and Behavioral Sciences: Sociology}{}
%\terms{Algorithms,  Economics}
%
%\keywords{network analysis, community detection, social networks, overlapping clusters}
%
%%\acmformat{M. F. Balcan and M. Braverman. Approximate Nash Equilibria in Perturbation Resilient Games.}

%\begin{bottomstuff}
%This work is supported by the National Science Foundation, under
%grant CNS-0435060, grant CCR-0325197 and grant EN-CS-0329609.
%
%Author's addresses: G. Zhou, Computer Science Department,
%College of William and Mary; Y. Wu  {and} J. A. Stankovic,
%Computer Science Department, University of Virginia; T. Yan,
%Eaton Innovation Center; T. He, Computer Science Department,
%University of Minnesota; C. Huang, Google; T. F. Abdelzaher,
%Computer Science Department, University of Illinois at Urbana-Champaign.
%\end{bottomstuff}

\maketitle

\section{Introduction}
%\vspace{-2mm}
\label{sec:intro}

\noindent {\bf Affinity Systems~} The problem of identifying
endogenously\footnote{endogenous: growing or developing from within; originating
  within.} formed overlapping communities or clusters
%The problem of endogenously determining possibly overlapping
%communities or clusters
arises in many contexts within e-commerce: finding overlapping communities in a social network, clustering retail products using collaborative filtering, clustering documents using citation information, classifying videos using viewing logs, etc. In such settings one needs to cluster the set of objects into meaningful, potentially overlapping subsets  by only using information about relations between
the objects. In this paper we develop the notion of an affinity system to model these scenarios.

 An affinity system
is a collection of elements with a set of ``preferences" each of these elements
has over other elements within the system. These preferences may be
expressed as a vector of  rankings, or, more generally, as a vector of non-negative weights representing affinities.
For example, when clustering videos,   affinities may represent the likelihood of the videos to be co-watched, with videos
that are co-watched more often ``ranking" each other higher.
When clustering documents, a document will ``prefer" documents it cites over documents it doesn't.

Perhaps the most natural application of affinity systems is to the study of social networks.
Social interaction is often determined by affinities among the members.
For example, in daily life, we often stay more in touch with people we like more.
When we go to a conference, we often hang out more with people with whom
we share more interests. Therefore, these social interactions, and their manifestations
as online social networks fit well within the affinity system paradigm.

\smallskip

\noindent {\bf  Endogenously Formed Communities in Affinity Systems~}
A central question concerning groups of individuals, documents, products, etc.,
is how to determine communities, or {\em overlapping} clusters that
capture the coherence among their members.
For example, in the context of retail products discussed above, it may be useful to automatically
``tag" the products with multiple categories for subsequent personalized marketing.
In the context of professional networks, a person may belong to multiple explicit  or implicit communities, for example
a scientist may simultaneously belong to the community of Economists and the community of Computer Scientist.
%may want to group products into categories, so that a product may belong to
%multiple categories. When viewing a product $X$, the user can then be presented with
%sponsored products that belong to one of the categories to which $X$ belongs.
The question of finding overlapping communities is closely related to the very well studied question of clustering~\cite{dudahart:book01},
but is much more general, since now elements may (and will) belong to multiple communities.
%
%In the present paper we argue that communities can be self-determined naturally by finding subsets in
%which each member is collectively preferred more by all other members of the subset
%than by any of those outside the subset, where preference is defined by the rankings or
%weights of the affinity system.

In this paper we formalize a natural notion of self-determined community and
 develop efficient algorithms to identify overlapping communities of this type as well as general bounds on the number of such communities.
Self-determined communities correspond to subsets that collectively
prefer each other more than they prefer those outside the subset,
where preference is defined by the rankings or
weights of the affinity system.
% Self-determined communities  correspond to subsets that
%are collectively preferred more by all other members of the subset
%than by any of those outside the subset, where preference is defined by the rankings or
%weights of the affinity system.
These communities are endogenously formed in the affinity system.
What is particularly nice about this formulation is that
we do not require that the subsets be of pre-specified sizes.  For example, a solution of
the flexible capacity roommate problem would group together people who prefer living
with each other to living with anyone else in another room. Switching to the context
of social and professional networks, an academic community can be viewed as a group
of scholars which appreciates the work of others in the community to that of the work
of people outside their community.
%The pairwise similarities in affinity systems for
%data sets can be used to tag the data into overlapping clusters.
In all these cases, the
overlapping communities or clusters are self-certified or self-determined.

\comment{
In the present paper we argue that communities can be self-determined
naturally by finding subsets in which each member is collectively preferred
more by all other members of the subset than by any of those outside the
subset, where preference is defined by the rankings or weights of
the affinity system.
What is particularly nice about this formulation
is that we do not require that the subsets be of pre-specified sizes.
Going back to the products example, a Lenovo laptop computer may be part of a cluster that contains ``laptops",
a bigger cluster that contains ``electronics", an a potentially overlapping smaller category of ``expensive Christmas gifts".
In this, and all other settings considered here, the overlapping
communities or clusters are self-certified or self-determined. We emphasize that no information beyond affinity data is needed
to reconstruct the overlapping clustering.

For example, a solution of the flexible capacity roommate problem would
group together people who prefer living with each other to living with
anyone else in another room.  Switching to the context of social and professional networks,  an academic community can be
viewed as a group of scholars which appreciates the work of others
in the community to that of  the work of people outside their community.
In all these cases, the overlapping
communities or clusters are self-certified or self-determined.
}

More formally, we study the mathematical structure of
   self-determined communities in an affinity system and
   design efficient algorithms for discovering them.
In our most basic model, we have $n$ members $V = \{1,...,n\}$
  in an affinity system, and we assume each member $i$ states
  a strict ranking $\pi_i$ of all members
  in the order of her preferences. %\footnote{For simplicity,
    % one may assume for all $i$, $\pi_i(1) = i$, i.e., everyone prefers herself more than anyone else.}.
%%We denote this affinity system by $A = (V,(\pi_1,...,\pi_n)$.
To evaluate whether a subset $S$ of size $|S| = k$ is a good community, imagine that
  each member $s\in S$ casts a vote for each of its $k$ most preferred members $\pi_s(1:k)$.
The number of votes that member $i$ receives,
   $\phi_S(i) = |\{i \in \pi_s(1:|S|) | s \in S\} | $, is the collective preference given by $S$.
We say $S$ is {\em self-determined} if everyone in $S$ receives
   more votes from $S$ than everyone outside $S$.

Different self-determined communities may have different degree of
coherence or robustness
  depending on both the fraction of votes received by the community members as well as
  the gap between the fraction of votes received by the
  community members and the non-community members.
To capture this, we say $S$ is a {\em $(\fracconnect,\alpha,\beta)$ self-determined community},
  for $0\leq \beta < \alpha \leq 1$ and $\fracconnect > 0$ if

%\vspace{-2mm}
\begin{enumerate}
%\vspace{-1mm}
%\setlength{\itemindent}{-4mm}
%\setlength{\itemsep}{-1mm}
\item [$\bullet$] each member $s\in S$ casts a vote for each of its $\fracconnect|S|$ most preferred members $\pi_s(1:\fracconnect|S|)$.
\item [$\bullet$] for each $i \in S$, the amount of vote $i$ receives, $\phi^{(\fracconnect)}_S(i) = |\{i \in \pi_s(1:\fracconnect|S|) | s \in S\} | $, is at least $\alpha|S|$.
\item [$\bullet$] for each $j \not\in S$, the amount of vote $j$ receives, $\phi^{(\fracconnect)}_S(j) = |\{j\in \pi_s(1:\fracconnect|S|) | s \in S\} | $, is at most $\beta|S|$.
\end{enumerate}
%\vspace{-2mm}

We start by analyzing how many communities can exist in an affinity system. Interestingly, we show that for constants $\alpha, \beta, \fracconnect$ we have a polynomial bound of
   $n^{O(\log(1/\alpha)/\alpha)}$ on the number of
   $(\fracconnect,\alpha,\beta)$-self-determined communities. Our analysis, using probabilistic methods, also yields a polynomial-time algorithm
  for enumerating these communities.
Moreover, we show that our bound is nearly tight, by exhibiting an  affinity
    system with $n^{\Omega(1/\alpha)}$
    $(\fracconnect,\alpha,\beta)$-self-determined communities.

%Inspired by the local graph clustering algorithm by Spielman and Teng \cite{SpielmanTengCuts},
We then present a local  community finding algorithm that is very efficient for an interesting range of parameters.
This algorithm, when given robustness parameters $\fracconnect, \alpha, \beta$, and a member $v \in V$, either returns
   a $(\fracconnect,\alpha,\beta)$-self-determined community of size $t$
   in time $O(f(\alpha,\beta,\fracconnect) \cdot t\log t)$ or  an empty set.
The algorithm satisfies the following performance guarantee: if $\alpha > 1/2$,
    if $v$ is chosen uniformly at random from a $(\fracconnect,\alpha,\beta)$-self-determined community  $S$, then with
   probability $\Omega(2\alpha-1)$, our local algorithm will successfully recover $S$ and so in time dependent only (and nearly) on the $|S|$
   and not on the size of the entire affinity system.
   As a consequence of this analysis, we can show that in the (natural)
   case when $\al>1/2$ we obtain a {\em near-linear} algorithm for finding all self-determined communities, substantially
   improving on the polynomial-time guarantee discussed above. Quasi-linear local algorithms are particularly important
   in the context of studying internet-scale networks, where even quadratic-time algorithms are not feasible, and where
   one sometime does not have access to the entire network but only to a local portion of it.
   The quasi-linear algorithm is one of our main technical contributions, as its techniques can potentially be used
   to convert other polynomial cluster-detection algorithms into local quasi-linear algorithms -- at least in the average case.

We also study {\em multi-facet affinity systems} where each member may have a
   number of different rankings of other members.
For example, member $i$ may have two rankings $\pi_{i,fun}$ and $\pi_{i,science}$,
  where first ranks members by how much fun $i$ thinks they are and the second
   ranks them according to academic affinity.
In this context, we say $S$ is a self-determined community if there exists a vector
of choices of rankings (in this case, in $\{fun, science\}^{|S|}$) such
that if members vote according to their associated choice, the
resulting votes self-certify $S$.
%   if when voting, each member first commits to one of
%  her rankings and then vote accordingly, and the resulting votes self-certify $S$.
We prove that if each member has a constant number of rankings, all our results can be extended,
   even though there could be exponential number of  combinations of rankings.

  %
%Some social networks, such Google+ with circles, provide more complex
%  information structure for recording social interactions.
%These types of social networks enable its users to share different things with different circles of people.

Our results can be extended to weighted affinity systems where the affinities of each member are given by a numerical
  weighting rather than just an ordinal ranking.
For example, member $i$ may give her most preferred member weight $1$,  next two preferred members weight  $0.7$,
  next one weight $0.5$, and so on. A weighted affinity system can be expressed as $A = \{V, a_1,...,a_n\}$, where $a_i$ is a $n$-dimensional vector $a_i = (a_{i,1},...,a_{i,n})$ and $0 \leq a_{i,j} \leq 1$ specifies the degree of affinity that $i$ has for $j$.
One can naturally define $(\fracconnect,\alpha,\beta)$-self-determined communities for weighted affinity systems.
The only requirement is that members are only allowed to cast votes up to a total weight of $\ta t$ when voting for a community of size $t$,
while respecting the affinity system.
%To satisfy this requirement, when voting for a community of size $t$, each member can cast a total fractional vote of $\fracconnect t$ in the order of her preference according to the weights; or by scaling down all the weights to make them add up to $\ta t$.
%if there are ties at the boundary,  %%\footnote{Why natural...}
 % she would naturally scale down the weights of those nodes just at the boundary to make the sum exactly equal to $\fracconnect|S|$.
We show that all our  bounds and algorithmic results extend to weighted affinity systems with only a slight loss in the parameters.
% by providing an efficient  reduction
%   from a weighted affinity system to a family of slightly larger ordinal affinity systems.

\smallskip

\noindent {\bf Endogenously Formed Communities in Social Networks~} Our general formulation enables us to shed light on the challenging task of defining and finding overlapping communities in social networks~\cite{modularity,mishra,mishra-conf,jure09,jure10}.
Typically, a social network can be viewed as graph $G=(V,E)$, where the edges could be either undirected
   (e.g., the Facebook social network determined by friend list) or directed (e.g. the Twitter network).
An edge could be unweighted or weighted (e.g., the Skype phone-call networks or the Facebook network based
on the number of times that one person writes on the wall of others).

It turns out that a social network can be realized as a projection of an affinity system.  Indeed, although our affinity systems are typically dense, their projections as social networks can be very sparse – as are many observed social networks.  We can think of our observed social network interactions as being induced (in various ways) by the underlying latent set of affinities.  To be precise, given a social network $G = \{V,W,w\}$ with weights $w =(w_{ij})$, we would like to recover the communities in the original affinity system.
%Consider a social network $G = \{V,E,w\}$, where $w$ defines the edge
%weights and without loss of generality, we assume $w(e)\in [0,1]$ (or
%$w(e) \in \{0,1\}$ if $G$ is an unweighted graph).
%One can think about a social network as a  potentially very sparse partial realization (or a projection) that we can
%observe of social interactions induced by some underlying latent set
%of affinities. Given such a social network $G = \{V,W,w\}$ with weights $w=(w_{ij})$,
%one would like then to recover the communities in the original
%affinity system.
  A natural way to do this is to lift the social
network back to an affinity system $A = \{V, a_1,...,a_n\}$ and then to solve the problem in
the lifted system.  For example, several natural approaches for
lifting based on different beliefs about how the social network may have emerged from an underlying set of affinities include:

% Given a social network $G = \{V,E,w\}$  with weights $w=(w_{i,j})$ there
%are several possible ways to define a corresponding affinity
%system $A = \{V, a_1,...,a_n\}$ based on different beliefs about how this social network
%may have emerged from an underlying set of affinities.
%Several natural approaches are:
% (see Section~\ref{discussion} for
%several examples).
%
%
% For example, given a social network $G
%= \{V, E, w\}$,
%\vspace{-2mm}
\begin{enumerate}
%\vspace{-1mm}
%\setlength{\itemindent}{-4mm}
%\setlength{\itemsep}{-1mm}
 \item {\em Direct
Lifting}: One can directly lift to an affinity system by defining $a_{i,j} = w_{i,j}$ if $(i,j) \in E$, otherwise $a_{i,j} = 0$
(we assume  WLOG that $w_{i,j}\in [0,1]$)).

\item {\em Shortest Path Lifting}: If $G=(V,E)$ is an unweighted
social network, and the shortest path distance from $i$ to $j$ is
$d_{i,j}$, one may define $a_{i,j} = 1/d_{i,j}$.  The shortest path
lifting can be extended to weighted cases by appropriated
normalization.

\item {\em Personal Page Rank Lifting}: Let  $p_i$ be the personal
PageRank vector \cite{AndersenChungLang} of vertex $i$, we define
$a_{i,j} = p_{i,j} /\max(p_i)$.

\item {\em Effective Resistance Lifting}:  Let $r_{i,j}$ be effective
resistance of from $i$ to $j$ by viewing $G$ a network of resistors,
using $1/w(e)$ as the resistance of $e\in E$ \cite{Doyle:book}, we
define $a_{i,j} = min_k(r_{i,k} /r_{i,j})$.
\end{enumerate}
%
%\footnote{We can also define a corresponding affinity system by using other
%quantities, such as hitting time and commute time in random walks
%\cite{Doyle:book} and diffusion-kernels \cite{STC:book04}.}
 Each style of lifting corresponds to a
particular belief on how this social network may have emerged from a
latent underlying affinity system.  For instance, Direct Lifting corresponds to the belief that
    a social network $G = (V,E)$, such as the Twitter network,
   arose from a latent affinity $A' = (V, \{a_1',..., a_n'\})$  by a process in which each
member $i$ connects to the $d_i$ top most elements %%twitter
   according to the affinity system $A'$.
   %system as follows:
%Suppose the attention capacity of individual $i$ is $d_i$.
%Denoting $D=\{d_1,...,d_n\}$, we let $G_A^D = \{V,E_A^D \}$ be the directed graph s.t.
%$E_A^D = \{ (i,j) | j \in \pi_{i}[1:d_i] \}.$
In other words, $i$ follows $j$, i.e., $(i,j)$ is a directed edge in this social network,
   if $j$ is among the $d_i$ top most elements of $i$ according to $A'$.
Similarly, one can think of
Shortest Path Lifting as corresponding to the belief that the social
network serves as an approximate spanner of the underlying
affinity system \cite{spanner:book}, and Effective Resistance Lifting
corresponds to the belief that a social network is
approximately based on some spectral sparsification of those underlying
affinities \cite{SpielmanTengPrecon}.

% Given a social network $G
%=(V,E,w)$, once we derive a corresponding affinity system $A$, we may
%use our notion of self-determined community and apply our algorithms
%and analysis to obtain communities in the original network.

Given a social network, %%$G =(V,E,w)$,
once we derive a corresponding affinity system $A$,
  we may use our notion of self-determined community and apply our algorithms and analysis
  to obtain communities in the original network.
%For example, for the direct lifting of $G$, the notion of the community we obtain is as follows:
% $S$ is a $(\fracconnect,\alpha,\beta)$-self-determined community in $G$ if every $i \in S$
%   receives at least $\alpha |S|$ collective vote and everyone not in $S$
%   receives at most $\beta |S|$ collective vote.
%If $G$ is unweighted, $(i,j) \in E$, $i\in S$, and $d_i$ is the out-degree of $i$,
%   then one way to set up the affinity system is to let  $i$ contribute $\min(1,\fracconnect |S|/d_i)$ to the collective vote of $j$.
From our analysis for affinity systems, we immediately obtain that there is a polynomial
   number of such communities in a social network, and they
   can be enumerated in polynomial time.

 We note that while the input social network is potentially very sparse, appropriate lifting procedures can produce an affinity system  better reflecting the true relationships between entities. Moreover,  many can be performed locally, allowing for our local algorithm to
determine meaningful communities especially efficiently. %% Our local algorithm

We also note that our study of multi-facet affinity systems allows us to model and analyze
   communities in more complex social networks -- such as  such Google+ with circles which enable its users to share different things with different circles of people. This extension may also enable us to model interdisciplinary sub-fields according to scientific works or interactions.

\smallskip

\noindent {\bf Self-determined Communities and $(\alpha,\beta)$-clusters~}  In this paper, we also provide several new results for communities defined as $(\alpha,\beta)$-clusters, a notion introduced by
Mishra, Schreiber, Stanton, and Tarjan~\cite{mishra-conf} for analyzing (unweighted) social networks.
In their definition,  $S$ is an $(\alpha,\beta)$-cluster (for $\alpha > \beta$) if
    for every $i\in S$, the number of neighbors coming
    from $S$ is at least $\alpha |S|$ and for every $j\not\in S$,
    the number of neighbors coming from $S$ is at most $\beta |S|$.
We prove that there exists a family of networks with superpolynomial
  number of $(\alpha,\beta)$-clusters.
For instance, if $\alpha = 1$ and $\alpha -\beta = 0.01$, then in $G(n,1/2)$, the Erd\"{o}s-Renyi random graph with $1/2$ edge probability, the expected number of $(\alpha,\beta)$-clusters is $n^{\Omega(\log n)}$.
We also show that under the assumption that the planted clique problem is hard, even finding a
            {\em single} $(\alpha,\beta)$-cluster
  is computational hard.
 Interestingly, our notion of communities in social networks obtained via direct lifting is quite similar
 to the notion of  $(\alpha,\beta)$-clusters, with the only difference that we bound the total amount of votes a member may cast.\footnote
 {For example, for the direct lifting of $G$, the notion of the community we obtain is as follows:
 $S$ is a $(\fracconnect,\alpha,\beta)$-self-determined community in $G$ if every $i \in S$
   receives at least $\alpha |S|$ collective vote and everyone not in $S$
   receives at most $\beta |S|$ collective vote.
If $G$ is unweighted, $(i,j) \in E$, $i\in S$, and $d_i$ is the out-degree of $i$,
   then one way to set up the affinity system is to let  $i$ contribute $\min(1,\fracconnect |S|/d_i)$ to the collective vote of $j$.}
This twist seems to be essential to obtain only a polynomial number of such communities
  and to be able to enumerate them in polynomial time.

%\smallskip
%\noindent {\bf Other Applications~} More generally, our algorithms could be used to  discover overlapping clusters in various data mining problems where this is desirable. These include categorizing documents or webpages in related topics, proteins by functions, generating non-tree-like subject classifications for documents, etc.

\smallskip
\noindent{\bf Related Work~} Problems of clustering or grouping data (based on network or pairwise similarity information or ranked data) have been extensively studied in many different fields. The classic goals have been to produce either a partition or a hierarchal clustering of the data~\cite{dudahart:book01,CGTS99,BBV08,cksw,bob07,jms06}. With the rise of online social networks, there has been significant recent interest in identifying overlapping clusters, or communities, in networks ranging
from professional contact networks to citation networks to
product-purchasing networks, with many heuristics and optimization criteria being proposed~\cite{jure09,jure10,modularity,mishra-conf,mishra,hop-11}.
However, much of this work has  disallowed natural communities such as those
containing highly popular nodes~\cite{SpielmanTengCuts,condkonst,jure09,jure10} or not given general guarantees on the number or
computation time needed to find all overlapping communities meeting
natural criteria~\cite{clsw10,mishra-conf,mishra,hop-11}. By contrast, our new formalization
 leads to natural communities and efficient algorithms
 for identifying all such communities. Additionally, our model allows us to deal with asymmetries in the input in a very natural way.

Independently, in recent work \cite{arora11} consider several assumptions (that are between worst case and average case)
concerning community structure and provide efficient algorithms in these settings.
%This yields a model that is a stochastic relaxation of the worst case
%model (that is NP-hard), and which can be solved efficiently.
Remarkably, while their setting is somewhat different from ours
some of their algorithms are similar in spirit. %% to ours.

%In recent work, Mishra, Schreiber, Stanton, and Tarjan~\cite{mishra-conf} have formalized a new graph clustering criterion that is particularly motivated for the problem of identifying natural overlapping clusters in social networks. As we show in our paper, their formulation allows for graphs with exponential number of communities.
%By contrast, our new formalization leads to a polynomial number of communities and efficient algorithms for identifying them.

%For clustering vertices in a network, a commonly used measurement of
%quality is the conductance of the set \cite{FriezeKannanVempala,LovaszSimonovitsFOCS,SpielmanTengCuts}, which measure the ratio of the
%number of edges leaving the set to the number of edges within the set.
%Although these definitions have found some success in practice, they
%  usually lead to clusters (communities) that are either non-overlapping
%  or sometimes counterintuitive \cite{???}. %%{\bf Need citation}

%\vspace{-3mm}
\section{Preliminaries and Notation}
%\vspace{-1mm}
In our most basic model, we consider an affinity system with $n$ members $V = \{1,...,n\}$ and
 assume that each member $i \in V$ states a strict ranking $\pi_i$ of all members  in the order of her preferences. Let $\Pi=\{\pi_1, \ldots, \pi_n\}$.
%We consider a natural notion of community structure in an affinity system over $n$ members, in which each member states her preference of
%all members in the system. We start by assuming that the preference of each member is a complete
%ranking of members.
For $t >0 $, $S \subseteq V$, $i \in V$ we denote by $\vn^t_S(i)$ the number of members in $S$ that place $i$ among the topmost
$t$ elements of their preference list. That is $\vn^t_S(i) = \sizeof{\setof{s\in S | i\in
    \pi_s(1:t)}}.$
For $\fracconnect >0$,  we let $\phi^{\fracconnect}_S(i):= \vn^{\lceil \fracconnect \sizeof{S}\rceil}_S(i)$. We define a natural notion of self-determined community as follows:
%\vspace{-1mm}
\begin{definition}
%\vspace{-1mm}
\label{prop:self-stable}
Given three positive parameters $\fracconnect, \alpha, \beta$, where $\beta < \alpha \leq 1$ and an affinity system $(V, \Pi)$ we say
that a subset $S$ of $V$ is an ($\fracconnect$, $\alpha$, $\beta$) {\em self-determined} community with respect to $(V, \Pi)$ if we have both
%\vspace{-2mm}
\begin{enumerate}
%\vspace{-1mm}
%\setlength{\itemindent}{-4mm}
%\setlength{\itemsep}{-1mm}
\item[(1)] For all $i\in S$, $\phi^\fracconnect_S(i) \geq \alpha \sizeof{S}$.
\item[(2)] For all $j\not\in S$, $\phi^\fracconnect_S(j) \leq \beta \sizeof{S}$.
\end{enumerate}
\end{definition}
%\vspace{-2mm}

 Throughout the paper, we will denote by $\gamma=\alpha-\beta$. %% and we  will usually assume that $\fracconnect$ is a constant.
   Fixing $\fracconnect$, we say that ``$i$ votes for $j$ with respect to a subset $S$'' if   $ j \in  \pi_i(1:\lceil \fracconnect \sizeof{S}\rceil)$.
When $S$ is clear from the context we say that $i$ votes for $j$.
%$i$ is in the top $\fracconnect \cdot t$ points in $j$'s list.

Note that communities {\em may overlap}. As a simple example, assume we have two sets $A_1$ and $A_2$ of size $n/2$ with $n/8$ nodes in common
(representing, say, researchers in Algorithms and researchers in Complexity). Assume each node in $A_i \setminus A_j$ ranks first the nodes in $A_i$ and then the nodes in $A_j$ and that each node in $A_i \cap A_j$ ranks the nodes in $A_i \cup A_j$ arbitrarily. Then each $A_i$ is a $(1, 3/4, 1/4)$ self-determined community.

We also  consider (more general) weighted affinity systems,  where the preferences of each member $i$ involve numerical weightings
(degrees of affinity) rather than just an ordinal ranking. A weighted affinity system is expressed as $A = \{V, a_1,...,a_n\}$,
where $a_i$ is a $n$-dimensional vector
$a_i = (a_{i,1},...,a_{i,n})$ and $0 \leq a_{i,j} \leq 1$ specifies the degree of affinity that $i$ has for $j$.
For example, $i$ may give her top-ranked node a weight of $1$,
 she might have a tie between its second and third-ranked nodes giving both a weight of $0.7$,
and so on.
If member $i$ chooses not to vote for a given node, this can be modeled by giving that node a weight of $0$.

We can naturally extend our notion of $(\fracconnect,\alpha,\beta)$-self-determined communities to weighted affinity systems. %% as follows.
In the definition of voting, the only requirement is that members are only allowed to cast votes up to a total weight of $\ta t$ when voting for a community of size $t$.
%while respecting the affinity system.
%
%When voting for a community of size $t$, $s$ can cast a total fractional
% vote of $\fracconnect t$ in the order of her preference according to the weights.
For example,  to evaluate whether a subset $S$ is a good community, %%imagine that
  each member $s \in S$ casts a weighted vote as follows: $s$ determines a prefix of the weights (sorted from highest to lowest)
   of total value $\fracconnect|S|$ and zeros out the rest. If there are ties at the boundary,  a natural %%\footnote{Why natural...}
   conversion is  to scale down the weights of those nodes just at the boundary to make the sum exactly equal to $\fracconnect|S|$.
In general, we denote the resulting vector (after capping the amount of vote a member casts when voting for a community of size $t$) as
$\w_{s}^{\fracconnect|S|}$.
  %%\footnote{We recover the model of the previous section if each member $s$ gives to its $k$th-ranked node a weight of $1-k\epsilon$ for infinitesimally-small $\epsilon$.} Let us denote by $w_s^{\fracconnect|S|}$ the induced weight vector.
The amount of the weight that member $i \in V$ receives from $S$ is $\w^{\fracconnect}_S(i) = \sum_{s\in S} \w_{s,i}^{\fracconnect|S|}$.
Given these, we can define an ($\fracconnect$, $\alpha$, $\beta$) {\em weighted self-determined} community as follows:

%\vspace{-1mm}
\begin{definition}
%\vspace{-1mm}
\label{prop:self-stable:weighted}
Given $\fracconnect, \alpha, \beta \geq 0$, $\beta < \alpha \leq 1$  and an weighted affinity system $(V, \W)$ we say
that a subset $S$ of $V$ is an ($\fracconnect$, $\alpha$, $\beta$) {\em weighted self-determined} community with respect to $(V, \W)$ if we have both
%\vspace{-2mm}
\begin{enumerate}
%\vspace{-1mm}
%\setlength{\itemindent}{-4mm}
%\setlength{\itemsep}{-1.5mm}
\item[(1)] For all $i\in S$, $\w^\fracconnect_S(i) \geq \alpha \sizeof{S}$.
\item[(2)] For all $j\not\in S$, $\w^\fracconnect_S(j) \leq \beta \sizeof{S}$.
\end{enumerate}
\end{definition}
%\vspace{-2mm}

We note that given an (weighted) affinity system and a set $S$ we can {\em test} in time polynomial in $n$ whether  a proposed set $S$ is a
$(\fracconnect,\alpha,\beta)$-self-determined community or not. Also, fixing a $(\fracconnect,\alpha,\beta)$-self-determined community $S$,
one can easily show that there exists a multiset $U$ of size $k(\gamma)=2 \log {(4 n)} /\gamma^2$ such that
the set of elements $i$ voted by at least a $(\alpha-\gamma/2)$ fraction of $U$ (or in the weighted  case, the set of elements $i$
receiving  $(\alpha-\gamma/2) |U|$ total vote from $U$)  is identical to $S$.
 This then implies a very simple quasi-polynomial procedure for finding all self-determined communities,
as well as an $n^{O(\log{n}/\gamma^2)}$ upper bound on the number of
$(\fracconnect,\alpha,\beta)$-self-determined communities. (See Appendix~\ref{additional} for details).

In this paper we present a {\em multi-stage} approach for finding an unknown community  in an affinity
system that provides much better guarantees for interesting settings of the parameters.
%which will provide much better guarantees for several interesting ranges of the parameters. %%$\fracconnect$.
%%what about the other parameters
At a generic level, this algorithm takes as input information $I$ about an unknown community $S$ and outputs a list  $\listt$
 of subsets of $V$ s.t. if  information $I$ is correct with respect to $S$, then with high probability $\listt$ contains $S$.
  This algorithm has two main steps: it first generates a list $\listt_1$ of sets $S_1$ s.t. at least one of the elements in $\listt_1$
  is a rough approximation to $S$ in the sense that $S_1$ nearly contains $S$ and it is not much larger than $S$. In the second step,
  it  runs a purification procedure to generate a list $\listt$ that contains $S$. (See Algorithm~\ref{alg:selfstable}.)
   Both steps have to be done with care by exploiting properties of self-determined communities and we will describe in detail in the following sections
ways to implement both steps of this generic scheme. We also discuss how to adapt this scheme  for outputting a self-determined community in a local manner, for enumerating all self-determined communities, as well as extensions to  multi-facet affinity systems and applications of our analysis to social networks.
%\vspace{-2mm}
\begin{algorithm}
  \caption{A generic algorithm for identifying an unknown community $S$}
  {\bf Input:}   Preference system $(V, \Pi)$, information $I$ about an unknown community $S$.
%\vspace{-1mm}
\begin{enumerate}
%\vspace{-1mm}
%\setlength{\itemindent}{-4mm}   %% move left 4 mm
%\setlength{\itemsep}{-1mm}      %% move items closer together 1mm
\item[($1$)] Using information $I$ to generate a list  $\listt_1$ of sets $S_1$ s.t. at least one of the elements in $\listt_1$ is a rough approximation to $S$.
\item[($2$)] Run a purification procedure to generate a list $\listt$ s.t. at least one of the elements in $\listt$ is identical to $S$.
  \item[($3$)] Remove from the list $\listt$ all the sets that are not self-determined  communities.
\end{enumerate}
{\bf Output:}  List of self-determined communities $\listt$.
\label{alg:selfstable}
\end{algorithm}
%\vspace{-1mm}

%both steps probabilistic in the local case; in the simple case only one probab.

%\vspace{-2mm}
\section{Finding Self-determined Communities}
%\vspace{-2mm}
\label{se:selfdeterminedalgo1}
In this section we show how to instantiate the generic Algorithm~\ref{alg:selfstable} if the information
we are given about the unknown community $S$ is its size and the parameters  $\fracconnect$, $\alpha$, and
$\beta$. We show that this leads to a polynomial time algorithm in the case where  $\fracconnect$, $\alpha$, and $\beta$ are constant.
We start with a  structural result showing that for any self-determined  community $S$ there exist a small number of community members
 s.t. the union of their votes contains almost all $S$.
%\vspace{-1mm}
\begin{lemma}
%\vspace{-1mm}
\label{multiple}Let $S$ be a $(\fracconnect,\alpha,\beta)$-self-determined  community.
Let $\gamma=\alpha-\beta$, $M ={\log{\left(16/\gamma\right)}}/{\alpha}$.
There exists a set $U$, $|U| \leq M$ s.t. the set $S_1= \{ i \in V| \exists s \in U, i \in \pi_s(1: \theta |S|) \}$
%%%  $i_1, \ldots, i_M \in S$ s. t. the union of their votes contains
satisfies $|S \setminus S_1| \leq (\gamma/16) |S|$.%%   $\geq 1-\gamma/16$ fraction of $S$.
\end{lemma}
%\vspace{-1mm}
\begin{proof}
%%%\vspace{-1mm}
%We start by showing that for any subset $U$ of $S$ there exists an element $i_U \in S$ which
%votes for at least $\alpha |U|$ of the members of $U$. We prove this by using double counting. Let $S=\{i_1, \ldots, i_t\}$ and $U=\{j_1, \ldots, j_u\}$, where $t=|S|$ and $u=|U|$.
%Consider a matrix $M$ with rows indexed by
%$j_1, ..., j_u$ and columns indexed  by $i_1, ..., i_t$. Set $M(j,i)$ to $1$ if $i$ votes for $j$.
%Since  $S$ be an $(\fracconnect,\alpha,\beta)$-self-determined  community, we must have $\alpha t$ entries set to $1$ in each row, so there
%must be a column with at least $\alpha u$ entries set to $1$, which then implies that
%there must exist a point $i_U$ such that votes for at least $\alpha |U|$ members of $U$.
Note that any subset $\tS$ of $S$ receives a total of at
least $\alpha|\tS||S|$ votes from elements of $S$, which implies that
for any such $\tS$ there exists $i_{\tS} \in S$ that votes for at least
$\alpha|\tS|$ members of $\tS$. Given this, we find the desired elements $i_1, \ldots, i_M \in S$ greedily one by one.
Formally, let $S_1=S$. Let  $i_1 \in S$ be an element that votes for at least a $\alpha |S_1|$ elements in $S_1$.
Let $S_2$ be the set $S$ minus the set of elements voted by $i_1$.
In general, at step $l \geq 2$,
%%by Lemma~\ref{single},
there exists $i_l \in S$ that votes by at least a $\alpha$ fraction of $S_l$. Let $S_{l+1}$  be the set $S_l$ minus
the set of elements voted by $i_l$. We clearly have $|S_{i+1}| \leq (1-\alpha)^i |S_1|$, so $|S_{M+1}| \leq (\gamma/16) |S_1|$ for $M ={\log{\left(16/\gamma\right)}}/{\alpha}$. By construction  the set $U=\{i_1, \ldots, i_M \in S\}$ satisfies the desired condition.
\end{proof}

Given Lemma~\ref{multiple}, we can  use the following procedure for generating a list that contains a rough approximation to $S$ which covers at least a $1-\gamma/16$ fraction of $S$ and whose size is at most ${\log{\left(16/\gamma\right)}}|S|$.

%\vspace{-2mm}
\begin{algorithm}
  \caption{Generate rough approximations}
  {\bf Input:}   Preference system $(V, \Pi)$, information $I$
  (parameters $\fracconnect$, $\alpha$, $\beta$, size $t$).
  %\vspace{-1mm}
\begin{enumerate}
%\vspace{-1mm}
%\setlength{\itemindent}{-4mm}   %% move left 4 mm
%\setlength{\itemsep}{-1mm}      %% move items closer together 1mm
\item[$\bullet$] Set $\listt=\emptyset$, $\gamma=\alpha-\beta$, $k_1(\fracconnect, \alpha, \gamma)= \log{\left(16/\gamma\right)}/{\alpha}$.
\item[$\bullet$] Exhaustively search over all subsets  $U$ of
$V$ of size $k_1(\fracconnect, \alpha,\gamma)$; for each set $U$  add to the  list $\listt$
the set $S_1 \subseteq V$ of points voted by at least an element in $U$ (i.e., $S_1= \{ i \in V| \exists s \in U, i \in \pi_s(1: \theta t) \}$).
\end{enumerate}
{\bf Output:}  List of sets $\listt$.
\label{alg:selfstablerough}
\end{algorithm}
%\vspace{-2mm}

We now describe a lemma that will be useful for analyzing the purification step, suggesting how we convert a rough approximation to $S$ into a list of candidate much-closer approximations to $S$.
%\vspace{-1mm}
\begin{lemma}
%\vspace{-1mm}
\label{almost}
Fix a $(\fracconnect,\alpha,\beta)$-self-determined  community $S$. Let  $\gamma=\alpha-\beta$, $t=|S|$, and $S_1 \subseteq V$, $|S_1| =M \fracconnect t$ s.t. $| S \setminus S_1| \leq \gamma t/16$.
Let $U$ be a set of $k$ points drawn uniformly at random from $\tilde{S} = S \cap S_1$.
Let  $S_2$ be the  subset of points in $S_1$   that are voted by  at least an
  $\alpha-\gamma/2$ fraction of nodes in $U$, i.e., $S_2 = \{ i \in S_1| \vn^{\fracconnect t}_{U}(i) \geq (\alpha-\gamma/2) |U|\}$.
If $k= 8 \log (32 \fracconnect M/\delta\gamma)/\gamma^2$, then with  probability $ \geq 1-\delta$, we have the symmetric difference $|\Delta(S_2, S)| \leq \gamma t/8$.
\end{lemma}
%\vspace{-1mm}
\begin{proof}
%%\vspace{-1mm}
We start by showing that the points in $\tilde{S}$ are voted by at least a $\gamma/2$ larger fraction of $\tilde{S}$ than the points in $S_1 \setminus \tilde{S}$.
Let $i \in \tilde{S}$. Since  $S$ is $(\fracconnect,\alpha,\beta)$-self-determined, at least  $\alpha t$ points in $S$ vote for $i$  and since $|S\setminus \tilde{S}| \leq \gamma t/16$ we get that at least $(\alpha-\gamma/16)t$ points in $\tilde{S}$ vote for $i$. Since $|\tilde{S}| \leq t$,
we obtain that at least a  $\alpha-\gamma/16$ fraction of points in $\tilde{S}$ vote for $i$.
Let $j$ be a point in $S_1 \setminus S$. We know that at most $\beta t$ points in $\tilde{S}$ vote for $j$ and since $|\tilde{S}| \geq (1-\gamma/16) t$, we have that at most a $\alpha-3 \gamma/4$ fraction of points in $\tilde{S}$ vote for $j$.

Fix $i \in S_1$. By Hoeffding's inequality, since $U$ is a set of $8\log (32 \fracconnect M/\delta\gamma)/\gamma^2$ points drawn uniformly at random from $\tilde{S}$ we have that
with probability  at least $1- \gamma \delta/(16 \fracconnect M) $ the fraction of points in $\tilde{S}$ that vote for $i$ is within $\gamma/4$ of the fraction of points in $U$ that vote for $i$.
%By Hoeffding we get that for any point $i$, the probability of error is at most $\gamma \delta/16 \fracconnectM$,
These together with the above observations imply that the expected size of $|\Delta(S_2, \tilde{S})|$ is $(\gamma \delta/(16 \fracconnect M )) \fracconnect M t= \gamma \delta t/16 .$
By Markov's inequality we obtain that there is at most a $\delta$ chance that $|\Delta(S_2, \tilde{S})| \geq \gamma t/16$.
Using the fact $|\tilde{S} \setminus S| \leq \gamma t/16$ we finally get that with probability $1-\delta$
we have $|\Delta(S_2, S)| \leq \gamma t/8$. %% as desired.
\end{proof}

%\vspace{-2mm}
\begin{algorithm}
  \caption{Purification procedure}
  {\bf Input:}   Preference system $(V, \Pi)$, information $I$ (parameters $\fracconnect$,
  $\alpha$, $\beta$, $\gamma$, $k_2(\fracconnect, \alpha,\gamma)$, $N_2(\fracconnect, \alpha,\gamma)$, size $t$),
  list of rough approximations $\listt_1$.
  %\vspace{-1mm}
\begin{enumerate}
%\vspace{-1mm}
%\setlength{\itemindent}{-4mm}   %% move left 4 mm
%\setlength{\itemsep}{-3mm}      %% move items closer together 1mm
\item[$\bullet$] For each element $S_1 \in \listt_1$, repeat $N_2(\fracconnect, \alpha,\gamma)$ times
%\vspace{-1mm}
  \begin{enumerate}
  %\vspace{-1mm}
  \setlength{\itemindent}{-8mm}   %% move left 4 mm
\setlength{\itemsep}{-1mm}      %% move items closer together 1mm
  \item[$\bullet$]  Sample a set $U_2$ of $k_2(\fracconnect, \alpha, \gamma)$ points at random from $S_1$.  Let $S_2 = \{ i \in S_1| \vn^{\fracconnect t}_{U_2}(i) \geq (\alpha-\gamma/2) |U_2|\}$.
       %%%i.e., $U_2$ is the subset of points in $S_1$ that are voted for by an at least $(\alpha-\gamma/2)$ fraction of nodes in $U_2$.
  \item[$\bullet$] Let $S_3 = \{ i \in V |\vn^{\fracconnect t}_{S_2}(i) \geq (\alpha-\gamma/2) |S_2|\}$.
  %%
%%  $S_3$  be the subset of points in $V$ that are voted for by an at least
     %%$(\alpha-\gamma/2)$ fraction of points in $S_2$.
     Add $S_3$ to the list $\listt$.
  \end{enumerate}
  \end{enumerate}
{\bf Output:}  List of sets $\listt$.
\label{alg:selfstablepurify}
\end{algorithm}
%\vspace{-2mm}

We now show how Lemmas~\ref{multiple} and~\ref{almost} can be used to identify and enumerate communities.
%\vspace{-1mm}
\begin{theorem}
%\vspace{-1mm}
\label{main}
Fix a $(\fracconnect,\alpha,\beta)$-self-determined  community $S$. Let
$\gamma=\alpha-\beta$, $k_1(\fracconnect, \alpha, \gamma)= \log{\left(16/\gamma\right)}/{\alpha}$,
 $k_2 (\fracconnect, \alpha, \gamma)= \frac{8}{\gamma^2} \log{\left(\frac{32 \fracconnect k_1}{\gamma \delta}\right)}$,
  $N_2(\fracconnect, \alpha, \gamma)=O({(\fracconnect k_1)}^{k_2} \log{(1/\delta)})$.
  Using Algorithm~\ref{alg:selfstablerough} together with Algorithm~\ref{alg:selfstablepurify} for steps (1) and (2) of
Algorithm~\ref{alg:selfstable}, %%we have that with probability at least...
%Using Algorithm~\ref{alg:selfstable} where we call procedure~\ref{alg:selfstablerough} with
% parameters $\fracconnect$, $\alpha$, $\beta$, $k_1(\fracconnect, \alpha, \gamma)$, and $t=|S|$ for generating a list of rough approximations,
% and procedure~\ref{alg:selfstablepurify}  with  parameters $\fracconnect$, $\alpha$, $\beta$,
%$k_2(\fracconnect, \alpha, \gamma)$, $N_2(\fracconnect, \alpha, \gamma)$, and $t=|S|$ for purification,
  we have that with probability $\geq 1-\delta$ one of the elements in the list $\listt$ we output is {\em identical} to $S$.
\end{theorem}
%\vspace{-1mm}
\begin{proof}
%%\vspace{-1mm}
Since when running Algorithm~\ref{alg:selfstablerough} we  search over all subsets of $U$ of $V$ of size
$k_1(\fracconnect, \alpha,\gamma)$,
by Lemma~\ref{multiple} in one of the rounds we find a set $U$ s.t. the set of points $S_1$  that are voted by at least an element in
$U$ cover a $1-\gamma/16$ fraction of $S$. So, $\listt_1$ contains a rough approximation to $S$.

Since $|S|=t$,  $U_2$ is a set of $k_2$ elements drawn at random from $\tilde{S}=S \cap S_1$ with probability $\geq (t/ (2 t \fracconnect k_1) )^{k_2}$. Therefore for $N_2=O((2 \fracconnect k_1)^{k_2} \log(1/\delta)) $,  with probability
$\geq 1-\delta/2$ in one of the rounds the set $U_2$ is a set
of $k_2$ elements drawn at random from $\tilde{S}$. In such a round, by Lemma~\ref{almost},
with probability  $\geq 1-\delta/2$ we get  a set $S_2$ such that $|\Delta(S_2, S)| \leq \gamma t/8$.
 A simple calculation  shows that $S_3=S$.
\end{proof}

%\vspace{-1mm}
\begin{corollary}
%\vspace{-1mm}
\label{number_communities}
The number of $(\fracconnect,\alpha,\beta)$-self-determined communities  in an affinity system $(V,\Pi)$ satisfies
 $B(n)=n^{ O\left(\log{\left(1/\gamma\right)}/{\alpha} \right)}  \left(\frac{\fracconnect  \log{\left(1/\gamma\right)}}{\alpha}\right)^{O\left(\frac{1}{\gamma^2} \log{\left(\frac{ \fracconnect \log{\left(1/\gamma\right)}}{\alpha \gamma}\right)}\right)}$ and with probability $\geq 1-1/n$  we can find all of them in time $B(n) \poly(n)$.
\end{corollary}

We note that Theorem~\ref{main} and Corollary~\ref{number_communities} apply even if some nodes do not list all members of $V$
in their preference lists, and then some nodes in a community $S$ have  fewer than $\fracconnect|S|$ votes in total.
If $\fracconnect$, $\alpha$, and $\beta$ are constant, then
Corollary~\ref{number_communities} shows that the number of communities is $O\left(n^{\log{\left(1/\gamma\right)}/{\alpha}} \right)$
which is {\em polynomial} in $n$ and they can be found in polynomial time.  We can show that the dependence on $n^{1/\alpha}$ is necessary: %Specifically:

%\vspace{-1mm}
\begin{theorem}
%\vspace{-1mm}
\label{lb}
For any constant $\fracconnect \geq 1$, for any $\alpha \geq 2\sqrt{\fracconnect}/n^{1/4}$,
there exists an instance such that the number of $(\fracconnect,\alpha,\beta)$-self-determined communities
 with  $\alpha-\beta=\gamma=\alpha/2$ is $n^{\Omega({1/\alpha})}$.
\end{theorem}
%\vspace{-1mm}
\begin{sketch}
%%%\vspace{-1mm}
Consider $L=\sqrt{n}$ blobs $B_1$, ..., $B_L$ each of size $\sqrt{n}$. %%Consider $\fracconnect=2$.
Assume that each point ranks the points inside its blob first (in an arbitrary order)
and it then ranks the points outside its blob randomly.
One can show that with non-zero probability for $l \leq n^{1/4}/(2\sqrt{\fracconnect})$ any union of  $l$ blobs  satisfies
the $(\fracconnect,\alpha,\beta)$-self-stability property with parameters $\alpha=1/l$ and $\gamma=\alpha/2$. Full details
 appear in Appendix~\ref{additional}.
\end{sketch}

%\medskip
\subsection{ Self-determined Communities in Weighted Affinity Systems}
%%\section{The Weighted Case}
\label{sec:wtd}

We provide here a simple efficient reduction from the weighted case to the non-weighted case.
%\vspace{-1mm}
\begin{theorem}
%\vspace{-1mm}
\label{reductionweighted}
Given a weighted affinity system $(V,\W)$,  $\fracconnect$, $\alpha$, $\beta$,  $\epsilon < \alpha$, and a community size $t$, there is an efficient
  procedure that constructs a non-weighted instance $(V',\Pi)$
along with a mapping $f$ from $V'$ to $V$, s.t.  for any
$(\fracconnect,\alpha,\beta)$ community $S$ in $V$  there exists a
$(\fracconnect,\alpha-\epsilon,\beta)$ community $S'$ in $(V',\Pi)$
with $f(S')=S$.
\end{theorem}
%\vspace{-1mm}
\begin{proof}
%%\vspace{-1mm}
Given   the original weighted instance $(V,\W)$, we construct a non-weighted instance $(V',\Pi)$  as follows.
For each  $\s \in V$, we create a blob $B_{\s}$ of $k$ nodes in $V'$. For any $\s, \ts \in V$,  if $p$ is the weight $\aff_{\s,\ts}^{\fracconnect t}$
with which $\s$ votes for $\ts$, we
connect $B_{\s}$ to $B_{\ts}$ with $G_{k,k, \lfloor pk \rfloor }$, where $G_{k,k, \lfloor pk \rfloor }$ is a bipartite graph with $k$ nodes on the left
and $k$ nodes on the right such that each edge on the left has out-degree $\lfloor pk \rfloor$ and each node on the right has in-degree $\lfloor pk \rfloor$.
 Clearly all nodes in $V'$ rank  at most $k|S|\fracconnect$ other nodes
 (and do not have an opinion about the rest).
%We can show that for $k=1/\epsilon$ we have: and for any nodes $\s$, $\ts$ in $V$, the number of nodes in $B_{\ts}$ that rank $\s$ belongs to
% $[\aff_{\ts,s}^{\fracconnect|S|} \cdot k - \epsilon k, \aff_{\ts,s}^{\fracconnect|S|} \cdot k + \epsilon k]$. Assume this happens.
Let $k=1/\epsilon$.
Consider a community $S$ in $(V,\W)$. For any $\s \in S$ and for each node in $i \in B_{\s}$ the total vote from nodes in $B_{\ts}$ for $\ts \in S$
(when evaluating whether  $\cup_{\ts \in S}B_{\ts}$ is a good community or not) is at least
 $\alpha |S| k - |S| \geq k|S| (\alpha -\epsilon)$. Moreover, for $\s \notin S$ and for each node in $B_{\s}$ we have the total vote
  from  the nodes in $B_{\ts}$ for $\ts \in S$
is at most $\beta|S| k$.
Therefore $\cup_{\ts \in S}B_{\ts}$  is a legal
$(\fracconnect, \alpha-\epsilon, \beta)$-self-determined community of size $kt$ in the non-weighted instance $(V',\W)$.
 %%  we get the desired result.
%Algorithmically, show we can run Algorithm \ref{alg:selfstable} on the resulting affinity system with
%the parameters $\fracconnect(1+\epsilon)$,$\alpha-\epsilon$,
%$\beta+\epsilon$ and $t$.
\end{proof}

Using this reduction we immediately get the following result:% (See Appendix~\ref{se:selfdeterminedalgo1} for a proof).

%\vspace{-2mm}
\begin{theorem}
%\vspace{-1mm}
\label{self-stable:weighted}
For any $\fracconnect$, $\alpha$, $\beta$, $\gamma=\alpha-\beta$,
the number of weighted $(\fracconnect,\alpha,\beta)$-self-determined communities is
 $B(n)=(n/\gamma)^{ O\left(\log{\left(1/\gamma\right)}/{\alpha} \right)}  \left(\frac{2\fracconnect  \log{\left(1/\gamma\right)}}{\alpha}\right)^{O\left(\frac{1}{\gamma^2} \log{\left(\frac{ \fracconnect \log{\left(1/\gamma\right)}}{\alpha \gamma}\right)}\right)}$
  and we can find them in time $B(n) \poly(n).$
\end{theorem}
\begin{proof}
%%\vspace{-1mm}
We perform the reduction in Theorem~\ref{reductionweighted} with $\epsilon=\gamma/2$ and use the algorithm in Theorem~\ref{main}
and the bound in Corollary~\ref{number_communities}.
The proof follows from the fact that  the number of vertices in the new instance
has increased by only a factor of $2/\gamma$.
 %% so the number of communities and the running time will still be polynomial.
We also note that each set output on the reduced instance
can then be examined on the original weighted affinity system, and kept iff it satisfies the community definition with
original parameters.
\end{proof}

%\vspace{-1mm}
%\begin{proof}
%%%\vspace{-1mm}
%We perform the reduction in Theorem~\ref{reductionweighted} with $\epsilon=\gamma/2$ and use the algorithm in Theorem~\ref{main}
%and the bound in Corollary~\ref{number_communities}.
%The proof follows from the fact that  the number of vertices %%in the new instance
%has increased by only a factor of $2/\gamma$.
% %% so the number of communities and the running time will still be polynomial.
%%%We also note that
%Each set output on the reduced instance
%can then be examined on the original weighted affinity system, and kept iff it satisfies the community definition with
%original parameters.
%\end{proof}
%%\vspace{-4pt}

%\medskip
%\noindent
\subsection{ Self-determined Communities in Multi-faceted Affinity Systems}

A multi-faceted affinity system is a system where each node may have more than one
rankings of other nodes.
%This may reflect, for example, that a person may have two rankings of
%other people, one corresponding to personal friends (in descending order of affinity), and one
%of co-workers.
Suppose that each element $i$ is allowed to have at most $f$ different rankings
$\pi_i^1,\ldots,\pi_i^f$. We say that the pair $(S,\psi)$ is a multi-faceted community where $\psi:S\rightarrow \{1,\ldots,f\}$,
if $S$ is a community where $\psi(i)$ specifies which ranking facet should be used by element $i$.
In other words, as before, let $\phi_{S,\psi}^\ta(i):=|\{s\in S|i\in \pi_s^{\psi(s)}(1:\lceil \ta |S|\rceil)\}|$. Then
$(S,\psi)$ is a $(\al,\be,\ta)$-multifaceted community if for all $i\in S$, $\phi_{S,\psi}^\ta(i)\ge\al |S|$, and for
all $j\notin S$, $\phi_{S,\psi}^\ta(j)<\be |S|$.

We show that for a bounded $f$, even though
 there may be exponentially many functions $\psi$, it is not harder to find multifaceted communities than to find regular communities.
Note that all our sampling
algorithms can be adapted as follows. Once a representative sample $\{i_1,\ldots,i_k\}$ of the  community $S$ is obtained,
we can guess the facets $\psi(i_1),\ldots,\psi(i_k)$ while adding a multiplicative $f^k$ factor to the running time. We can thus
get the set $S_2$ approximating $S$ in the same way as it is found in Algorithms \ref{alg:selfstablerough} and \ref{alg:selfstablepurify}
while adding a multiplicative factor of $f^{k_1+k_2}$ to the running time. We thus obtain a list $\listt$ that for each multi-faceted community $(S,\psi)$
contains a set $S_2$ such that $\Delta(S_2,S)<\ga t/8$. Given $S_2$ we can output $S$ with probability $>f^{-8 \log n / \ga^2}/2$:
 guess a set $U_2$ of $m=8\log n /\ga^2$ points in $S_2$; guess a function $\psi_2$ on $U_2$; output
$S=$ the set of points that receive at least $(\al-\ga/2)t$ votes according to $(U_2,\psi_2)$.  Moreover,
a facet structure $\psi'$ can be recovered on $S$ so that $(S,\psi')$ is an $(\al-\ga/4,\be+\ga/4,\ta)$-multifaceted community using
a combination of linear programming and sampling. Details appear in Appendix~\ref{appendix:facet}.

%\vspace{-1mm}
\begin{theorem}
%\vspace{-1mm}
\label{thm:mutlf}
Let $S$ be an $f$-faceted $(\al,\be,\ta)$-community. Then there is an algorithm that runs in $O(n^2)$ time and outputs $S$, as well as a facet structure $\psi'$ on $S$ such that
$(S,\psi')$ is an $(\al-\ga/4,\be+\ga/4,\ta)$-multifaceted community with probability
$ \geq
(f\cdot n)^{- O\left(\log{\left(1/\gamma\right)}/{\alpha} \right)}  \left(\frac{f\cdot \fracconnect  \log{\left(1/\gamma\right)}}{\alpha}\right)^{-O\left(\frac{1}{\gamma^2} \log{\left(\frac{ \fracconnect \log{\left(1/\gamma\right)}}{\alpha \gamma}\right)}\right)} f^{-O(\log n/\ga^2)}.
$
\end{theorem}

%\vspace{-3mm}
\section{A Local Algorithm for Finding Self-determined Communities}
%\vspace{-1mm}
\label{sec:local}
In this section we describe a local algorithm for finding a community.
Given a single element $v$ and the target community size $t$, the goal of the
 algorithm is to output a community $S$ of size $t$ containing $v$.
Let us fix a target community $S$ that we are trying to uncover this way.

We  note that we need $\al>1/2$ for a local algorithm that uses only one seed to succeed. If $\al \le 1/2$ then one
may have a valid $(\ta,\al,\be)$-community that is comprised of two disjoint cliques of vertices. In this case,
no local algorithm that starts with just one vertex as a seed may
uncover {\em both} cliques, however we can extend the construction
below if we start with  $O(1/\alpha)$ seeds.
Below, we focus on providing a local algorithm for $\al>1/2$. Our local algorithm will follow the structure of the generic Algorithm~\ref{alg:selfstable}.
The main technical challenge is to provide a local procedure for producing rough approximations.
In general, it is not possible to do so starting from
{\em any} seed vertex $v\in S$. For example, if $v$ is a super-popular vertex that is voted first by {\em everyone} in $V$, then
$v$ will belong to all communities including $S$, but $v$ would contain no ``special information" that would allow one to identify
$S$. However, we will show that {\em a constant fraction} of the nodes in $S$ are sufficiently ``representative" of $S$ to
enable one to recover $S$.

Let us fix $t$ and $\fracconnect$. For an element $v$, we let $R(v)$ be a uniformly
random element which receives $v$'s vote with these parameters. In other words,
$
R(v):=\text{uniform element of }~\pi_v(1:\ta \cdot t).
$ We start  with the main technical claim that enables a local procedure for producing rough approximations.
%\vspace{-1mm}
\begin{lemma}
%\vspace{-1mm}
\label{cl:1}
Let $S$ be any $(\ta,\al,\be)$-community of size $t$. Let $\eta:=2\al-1>0$. Then there is a subset $T\subseteq S$ such that
$|T|\ge \eta t$ and for each pair $v\in T$ and $u\in S$, we have
$
\Pr[R(R(v))=u] \ge \frac{(\al-1/2)/\ta^2}{t}.$
\end{lemma}%In other words, if we start with any element $v$ of $T$ as our ``seed", take a random element $w$
%on $v$'s list and then take a random element on $w$'s list, then each member of $S$ has a $\Theta(1/t)$ probability to be selected.
%Thus this selection procedure may be used to obtain a uniformly random element of $S$, at least with some constant probability.
%\vspace{-1mm}
\begin{proof} %%(of Claim~\ref{cl:1}).
%%\vspace{-1mm}
For each element $v\in S$ denote by $O_S(v):= \pi_v(1:\ta\cdot t)\cap S$ -- the elements of $S$ that $v$ votes for, and by
$I_S(v):= \{u\in S:~v\in \pi_u(1:\ta\cdot t)\}$ -- the elements of $S$ that vote for $v$. By the community property we know that
$|I_S(v)|\ge \al t$ for all $v\in S$. %%In general, we know nothing about $|O_S(v)|$.
Observe that
$$
\sum_{v\in S} |O_S(v)| = \sum_{v\in S}|I_S(v)| \ge \al t^2.
$$
Hence at least an $\eta$-fraction of $v$'s in $S$ must satisfy $|O_S(v)|\ge t/2$, where $\eta=2\al-1$. Let
$
T:=\{v:~|O_S(v)|\ge t/2\}\subseteq S.
$
For any $v\in T$ and any $u\in S$, we have $$|O_S(v)\cap I_S(u)|\ge |O_S(v)|+|I_S(u)|-t \ge (\al-1/2)\cdot t.$$
%Denote $W:=O_S(v)\cap I_S(u)$.
To finish the proof note that
$$
\Pr[R(R(v))=u] \ge \Pr[R(v)\in O_S(v)\cap I_S(u)] \cdot \frac{1}{\ta \cdot t} \ge \frac{(\al-1/2)\cdot t}{\ta\cdot t} \cdot \frac{1}{\ta \cdot t} =
\frac{(\al-1/2)/\ta^2}{t}.
$$
\end{proof}

%%%%%%%%%%%%%%%%%%%%%%%%%%%%%%%%%%%%%%%%%

We call any vertex $v$ in the set $T$  in Lemma~\ref{cl:1} a ``good seed vertex'' for $S$.
Lemma~\ref{cl:1} suggests a natural procedure (Algorithm~\ref{alg:selfstableroughlocal}) for generating a rough approximation
in a local way given a good seed vertex.
%We are now ready to present Algorithm~\ref{alg:loc} which is the local version of Algorithm~\ref{alg:selfstable}

%\vspace{-2mm}
\begin{algorithm}
  \caption{Generate rough approximations}
  {\bf Input:}   Preference system $(V, \Pi)$, information $I$ (parameters $\fracconnect$, $\alpha$, $\beta$, $\gamma$, vertex $v$, size $t$).
  %\vspace{-1mm}
\begin{enumerate}
%\vspace{-1mm}
\item[$\bullet$] Set $S_1=\left\{u:~\Pr[u=R(R(v))]\ge \frac{(\al-1/2)/\ta^2}{t}\right\}$.
\end{enumerate}
{\bf Output:}  List of sets $\listt=\{S_1\}$.
\label{alg:selfstableroughlocal}
\end{algorithm}
%\vspace{-1mm}

%\vspace{-1mm}
\begin{theorem}
%%\vspace{-1mm}
\label{thm:loc1}
Assume $\al>1/2$.
Let $k_2(\ta,\al,\ga)=O\left(\frac{\log(\ta/\de \ga(\al-1/2))}{\ga^2}\right)$,
 $N_2(\ta,\al,\ga)=\left(\frac{\ta^2}{\al-1/2} \right)^{k_2(\ta,\al,\ga)}\log (1/\de)$.
 Assuming $v$ is a good seed element for a community $S$, then by using Algorithm~\ref{alg:selfstableroughlocal} together with
Algorithm~\ref{alg:selfstablepurify} for steps (1) and (2) of
Algorithm~\ref{alg:selfstable}, we have that with probability $\geq 1-\de$
we will output  $S$.
\end{theorem}
%\vspace{-1mm}
\begin{proof}
%%\vspace{-1mm}
It is enough to show that each iteration of the purification algorithm (Algorithm~\ref{alg:selfstablepurify}) has a probability
$\geq \left(\frac{\al-1/2}{\ta^2} \right)^{k_2}$ to output $S$.
Since $v$ is a good  seed element of $S$, the set $S_1$ produced by Algorithm~\ref{alg:selfstableroughlocal} must contain $S$.
It is easy to see that  $|S_1| \leq t\ta^2/(\al-1/2)$. Thus, applying Lemma~\ref{almost} with $M=\ta/(\al-1/2)$ we see that if the points of $U_2$ are drawn uniformly
from $S$, then with high probability $S_2$ is $\ga/8$-close to $S$, and $S_3=S$. Since conditioned on $U_2\subseteq S$, $U_2$ is
uniform in $S$, our probability of success is given by the probability that $U_2\subseteq S$, which is equal to
$
\left( \frac{|S|}{|S_1|} \right)^{k_2} \ge
%\left( \frac{t}{t\ta^2/(\al-1/2)}
%\right)^{k_3} =
\left( \frac{\al-1/2}{\ta^2}
\right)^{k_2}
$,
which completes the proof.
\end{proof}

Note that when $\al>1/2$, $\be$, and $\ta$ are constants, the purification procedure will run in a constant number of iterations.
Our main result of this section is the following:
%(See Appendix~\ref{addsec:local} for a proof).

%\vspace{-1mm}
\begin{theorem}
%\vspace{-1mm}
\label{thm:mainloc}
Suppose $\alpha > 1/2$.  Assume $\alpha$, $\beta$, $\fracconnect$, and $\delta$ are
constants.  If $v$ is chosen uniformly at random from $S$, then with
probability at least $(2\alpha-1)(1-\delta)$ we can find $S$ in time
$O(t \log t)$.
\end{theorem}
%
%Theorem~\ref{thm:mainloc} implies that if we choose the seed vertex randomly and uniformly from the entire vertex set $V$, then with probability
%$\Omega(t/n)$ we will recover $S$ in time $O(t\log t)$. Thus given $t$ this immediately gives an algorithm for finding $S$ with a constant probability in time $O(n\log t)$.

\begin{proof}
First, by Lemma 3, with probability at least $2\alpha-1$, element $v$
is such that for all $u \in S$, we have $\Pr[R(R(v)) =u] \geq
\frac{\alpha-1/2}{\fracconnect^2t}$.  We now implement Algorithm 4 by
performing $\left(\frac{8\fracconnect^2t}{\alpha-1/2}\right)\log(2t/\delta)$
random draws from $R(R(v))$ and letting $S_1$ be the set of points $u$
hit at least $4\log(2t/\delta)$ times.   By Chernoff bounds, for each
$u \in S$, we have included $u$ in $S_1$ with probability at least $1
- e^{-8\log(2t/\delta)/8} = 1-\delta/(2t)$, so with probability at
least $1-\delta/2$ we have $S_1 \supseteq S$.   Furthermore, since we
only include points hit at least $4\log(2t/\delta)$ times, we have
$|S_1| \leq \left(\frac{2\fracconnect^2 t}{\alpha - 1/2}\right)$.  Thus, the
analysis in Theorem 3 implies that the purification step (Algorithm 3)
will succeed with probability at least $1 - \delta/2$ for a choice of
$N_2 = \left(\frac{2\fracconnect^2}{\alpha -
1/2}\right)^{k_2(\fracconnect,\alpha,\gamma)} \log (2/\delta)$.  Putting
these together yields the desired success probability.  Furthermore,
since $\alpha$, $\beta$, $\fracconnect$, $\delta$ are constants, the overall
time is $O(t \log t)$.
\end{proof}

It is not hard to see that
the algorithm in Theorem~\ref{thm:loc1} will work even if $t$ is given to it only up to some small multiplicative error.
%%As a consequence of Theorem~\ref{thm:loc} we obtain a quasi-linear bound on the number of possible communities.
As a corollary of Theorem~\ref{thm:loc1}, we see that the number of communities is actually linear and we can find all of them in quasilinear time.
%%Moreover, [we can find them all... depends on the definition of ``find"...]
%(See Appendix~\ref{addsec:local} for a proof.)
%\vspace{-1mm}
\begin{theorem}
%\vspace{-1mm}
\label{quasilinearbound}
Suppose that $\al>1/2$. The total number of $(\ta,\al,\be)$-self-determined communities is bounded by
$
O\left(
n \cdot\frac{1}{\min(\ga,1/2-\al)} \cdot  \left(\frac{\ta^2}{\al-1/2} \right)^{O\left(\frac{\log(\ta/\de \ga(\al-1/2))}{\ga^2}\right)}
\right),
$
which is $O(n)$ if $\al$, $\be$, and $\ta$ are constants.
\end{theorem}

\begin{proof}
It is easy to see that executing the Algorithm in Theorem~\ref{thm:mainloc} where we only do one iteration of the purification step (i.e., of Algorithm~\ref{alg:selfstablepurify})  with
inputs $t'\in ((1-\ve)t,(1+\ve)t)$, $\al'=\al-4 \ve$, $\be'=\be+4 \ve$, $\ta'=\ta(1+\ve)$, and an appropriate seed
vertex $v\in S$ will lead to a discovery of an $(\ta, \al,\be)$-community of size $|S|=t$ with probability
 $\ge p:= \left(\frac{\al-1/2}{\ta^2} \right)^{k_2(\theta,\alpha,\gamma)}$, as long as $\ve$ is sufficiently small.
Here it is enough to take $\ve = \min(\ga, \al -1/2)/100$.  Thus a pair $(v,t')$, where $v$ is a vertex and $t'$ is the target size
corresponds to at most $1/p$ distinct communities.
Moreover, each community $S$ of size
$t$ corresponds to more than $t(2\al-1)/2$ such pairs.
Since $t'$ needs only to be within a multiplicative $(1+\ve)$ from $t$,
we can always select $t'$ from the set of values $\{(1+\ve)^i:~i=0,1,\ldots, \lceil \log_{1+\ve}n \rceil\}$. For each
value $t'$, the number of communities of size between $t'$ and $t'(1+\ve)$ is thus bounded by the number of possible pairs $(t',v)$ ($=n$),
times $1/p$ and divided by $t'(2\al-1)/2$:
$$
\#\{\text{communities of size between $t'$ and $t'(1+\ve)$}\} \le
\frac{n}{t'} \cdot \frac{1/p}{(2\al-1)/2}.
$$
Summing over the possible values of $t'$ we obtain the upper bound:
$$
n \cdot \frac{2}{\ve (2\al-1)} \cdot\left(\frac{\ta^2}{\al-1/2} \right)^{k_2(\theta,\alpha,\gamma)},
$$
which leads to the bound in the statement of the theorem.
\end{proof}

\noindent{\bf Note:~} We can extend our local approach  to  weighted and multi-faceted affinity systems. See Appendix~\ref{addsec:local}.

%See Appendix~\ref{alternative-non-local}. % for details.

\subsection{An Alternative Non-local Algorithm}
\label{alternative-non-local}
%\vspace{-1mm}
The analysis in  this section suggests an alternative way for generating rough approximations in the non-local model which leads to an algorithm that provides asymptotically better bounds than Theorem~\ref{main}
in interesting cases, in particular when $\theta$, $\alpha$, and $\gamma$ are
constants and there is a large gap
between $\alpha$ and $\gamma$. This leads to an improved polynomial bound of
   $n^{O(\log(1/\alpha)/\alpha)}$ on the number of
   $(\fracconnect,\alpha,\beta)$-self-determined communities when $\theta$, $\alpha$, and $\gamma$ are
constants using Algorithm~\ref{alg:selfstablerough1}:

%\vspace{-1mm}
\begin{algorithm}
  \caption{Generate rough approximations}
  {\bf Input:}   Preference system $(V, \Pi)$, information $I$
  (parameters $\fracconnect$, $\alpha$, $\beta$, size $t$).
  %\vspace{-1mm}
\begin{enumerate}
%\vspace{-1mm}
%\setlength{\itemindent}{-4mm}   %% move left 4 mm
%\setlength{\itemsep}{-1mm}      %% move items closer together 1mm
\item[$\bullet$] Set $\listt=\emptyset$; $\gamma=\beta-\alpha$.
\item[$\bullet$] Exhaustively search over all subsets  $U_0$ of
$V$ of size $\lceil (\log 1/\al)/\al\rceil+1$; for each $U_0$ to the $\listt$  the set
 $S_1:= \left\{x: \sum_{y\in U_0} \Pr[x=R(R(y))]\ge \frac{\al }{2\ta^2 t}\right\}$.
\end{enumerate}
{\bf Output:}  List of sets $\listt$.
\label{alg:selfstablerough1}
\end{algorithm}
%\vspace{-2mm}

%\vspace{-1mm}
\begin{theorem}
%\vspace{-1mm}
\label{main1}
Fix a $(\fracconnect,\alpha,\beta)$-self-determined  community $S$. Let
$\gamma=\alpha-\beta$, $k_1(\fracconnect, \alpha, \gamma)= O\left(\log{\left(1/\alpha\right)}/{\alpha} \right)$,
 $k_2 (\fracconnect, \alpha, \gamma)= O\left(\frac{1}{\gamma^2} \log{\left(\frac{ \fracconnect k_1}{\gamma \delta}\right)} \right)$,
  $N_2(\fracconnect, \alpha, \gamma)=O({(\fracconnect^2/\alpha^3)}^{k_2} \log{(1/\delta)})$.
  Using Algorithm~\ref{alg:selfstablerough1} together with Algorithm~\ref{alg:selfstablepurify} for steps (1) and (2) of
Algorithm~\ref{alg:selfstable}, %%we have that with probability at least...
%Using Algorithm~\ref{alg:selfstable} where we call procedure~\ref{alg:selfstablerough} with
% parameters $\fracconnect$, $\alpha$, $\beta$, $k_1(\fracconnect, \alpha, \gamma)$, and $t=|S|$ for generating a list of rough approximations,
% and procedure~\ref{alg:selfstablepurify}  with  parameters $\fracconnect$, $\alpha$, $\beta$,
%$k_2(\fracconnect, \alpha, \gamma)$, $N_2(\fracconnect, \alpha, \gamma)$, and $t=|S|$ for purification,
  then with probability $\geq 1-\delta$ one of the elements in the list $\listt$ we output is {\em identical} to $S$.
\end{theorem}

\begin{sketch}
By using a reasoning similar to the one in Lemma~\ref{almost} we can show that there exist a set $U_0$ of $ \lceil (\log 1/\al)/\al\rceil+1$
points  such that the subset $U_1$ of points voted by at least a member in $U_0$ contains  $\geq 1-\al/2$ fraction of $S$.
We show in the following that the corresponding set $S_1$ indeed covers $S$. Fix a vertex $x\in S$. We need to show that
$$\sum_{y\in U_0} \Pr[x=R(R(y))]\ge \frac{\al }{2\ta^2 t}.$$
%\begin{equation}
%\label{eq:T0}
%\sum_{y\in U_0} \Pr[x=R(R(y))]\ge \frac{\al }{2\ta^2 t}.
%\end{equation}
Let $Q\subseteq S$ be the set of elements that vote for $x$. We know that $|Q|\ge \al t$, since $x\in S$. Thus
$$|U_1\cap Q| \ge |U_1|+|Q|-|S| > \al t/2.$$
Each $z\in U_1\cap Q$ contributes at least $1/\ta^2 t^2$ to the sum  $\sum_{y\in U_0} \Pr[x=R(R(y))]$. Thus this sum  is
at least
$
(\al t/2)\cdot (1/\ta^2 t^2) = \al/(2\ta^2 t)
$.
Hence $x\in S_1$, as required. Moreover, by observing that
$$
\sum_x \sum_{y\in U_0} \Pr[x=R(R(y))] = \sum_{y\in U_0} \sum_x \Pr[x=R(R(y))] < 1/\al^2,
$$ we obtain $|S_1|< \frac{2\ta^2 t}{\al^3}.$

Since when running Algorithm~\ref{alg:selfstablerough1} we exhaustively search over all subsets of $U_1$ of $V$ of size
$k_1(\fracconnect, \alpha,\gamma)$, in one of the rounds we find a set $U_1$ s.t. $|S_1|< \frac{2\ta^2 t}{\al^3}$, $S \subseteq S_1$.
So, $\listt_1$ contains a rough approximation to $S$. Finally, using a reasoning similar to the one in Theorem~\ref{main} we get the desired conclusion.
%Since $|S|=t$, there is $\geq (t/ (\frac{2\ta^2 t}{\al^3}) )^{k_2}$ chance
%that $U_2$ is a set of $k_2$ elements drawn at random from $\tilde{S}=S \cap S_1$, thus
%for   $N_2(\fracconnect, \alpha, \gamma)=O({(\fracconnect/\alpha)}^{k_2} \log{(1/\delta)})$ we have that with probability
%$\geq 1-\delta/2$ in one of the rounds the set $S_2$ is a set
%of $k_2$ elements drawn at random from $\tilde{S}$. In such a round, by lemma~\ref{almost},
%with probability  $\geq 1-\delta/2$ we get  a set $S_2$ s. t. $|\Delta(S_2, S)| \leq \gamma t/8$.
% A simple computation then shows that $S_3=S$.
\end{sketch}

Theorem~\ref{main1} %%(proven in Appendix~\ref{addsec:local})
 gives asymptotically better bounds than Theorem~\ref{main}
when $N_1 = n^{k_1(\theta,\alpha,\gamma)}$ is the
dominant term in the bound (e.g., when $\theta$, $\alpha$, and $\gamma$ are
constants) and especially when there is a large gap
between $\alpha$ and $\gamma$ -- since $k_1$ is reduced from
$\log(16/\gamma)/\alpha$ to $\lceil\log(1/\alpha)/\alpha\rceil+1$.  On
the other hand, Theorem~\ref{main1} has worse dependence on $\theta$ and
$\alpha$ in $N_2$, so for certain parameter settings, Theorem~\ref{main} can be
preferable especially if one optimizes the constants in Lemmas~\ref{multiple} and~\ref{almost}
based on the given parameters.

%\vspace{-3mm}
\section{Self-determined Communities in Social Networks}
\label{sec:sn}
%\vspace{-1mm}
%%Connections to the notion of $(\alpha,\beta)$-clusters~\cite{mishra-conf,mishra}
In this section we present a natural notion of self-determined communities in social networks and discuss how our analysis sheds light on the notion of $(\alpha, \beta)$-clusters~\cite{mishra-conf,mishra,hop-11}. We assume that the input is a directed graph $G = (V,E)$ and  %%as it is often the case in the social networks literature~\cite{}.
for a vertex $i$ we denote by $d_i$ its out-degree.
As discussed in Section~\ref{sec:intro}, given a social network we can consider the affinity system induced by direct lifting and then
consider self-determined communities in that affinity  system. This leads to the following very natural notion: %%% of communities. %% in directed graphs.
%\vspace{-1mm}
\begin{definition}
%\vspace{-1mm}
\label{prop:self-stable:graph1}
 Let $G = (V,E)$
be a directed graph and let $\fracconnect, \alpha, \beta \geq 0$ with $\beta < \alpha \leq 1$.  Consider the affinity system $(V, {a_1, \ldots, a_n})$ where $a_{i,j} = w_{i,j}$ if $(i,j) \in E$ and  $a_{i,j} = 0$ otherwise.
A subset  $S \subseteq V$ is a $(\fracconnect, \alpha,\beta)$ self-determined community in $G$ if it is a
  $(\fracconnect,\alpha,\beta)$ weighted self-determined community in $(V, {a_1, \ldots, a_n})$.
\end{definition}
%\vspace{-2mm}

Note that when evaluating a community of size $t$ each
node $i$ is allowed a total vote of at most $\theta t$. %Moreover,
% and as a consequence we only fractionally count edges from
%high-degree nodes. In particular,
% due to the way we handle ties in weighted affinity systems, we
One natural way to achieve this is to only fractionally count edges from
high-degree nodes $i$, giving them weight $\min(\fracconnect t/d_i,1)$ when evaluating a community of size $t$ in the induced weighted affinity system.

The  community notion introduced in~\cite{mishra-conf,mishra} is as follows:
%\vspace{-2mm}
\begin{definition}
%\vspace{-1mm}
\label{MSST} Let $\alpha, \beta$ with $\beta < \alpha \leq 1$ be two positive parameters.
Given an undirected graph, $G = (V,E)$, where every vertex has a self-loop, a subset  $S \subseteq V$
is an $(\alpha, \beta)$-cluster if $S$ is:
%\vspace{-1mm}
\begin{enumerate}
%\vspace{-1mm}
%\setlength{\itemindent}{-4mm}
%\setlength{\itemsep}{-1mm}
\item[(1)] Internally Dense: $\forall i \in S$, $|E(i,S)| \geq \alpha |S|$.
\item[(2)]  Externally Sparse: $\forall i \notin S$, $|E(i,S)| \leq \beta |S|$.
\end{enumerate}
\end{definition}
%\vspace{-2mm}

%%%For this notion,~\cite{mishra-conf,mishra} provided an efficient deterministic algorithm
The $(\alpha, \beta)$-cluster  notion resembles our community notion in
Definition~\ref{prop:self-stable:graph1}. In particular, in the case
where the graph is undirected, Definition~\ref{prop:self-stable:graph1} is
similar to Definition~\ref{MSST}, except that in the case of our Definition~\ref{prop:self-stable:graph1} each node $i$ is allowed a total vote of at most $\theta t$. As discussed above one way to achieve this is to only fractionally count edges from
high-degree nodes $i$, giving them weight
$\min(\fracconnect|S|/d_i,1)$.
This distinction is
crucial for getting polynomial time algorithms.
From our results in the previous sections we have that every graph has only a polynomial number of communities satisfying Definition~\ref{prop:self-stable:graph1} and
moreover, we can find all of  them in polynomial time.
In contrast, as we show, there exist graphs with a superpolynomial number of $(\alpha, \beta)$-clusters.

%We begin with the first result.

%\vspace{-1mm}
\begin{theorem}
%\vspace{-1mm}
\label{lb-mishra-com}
For any constant $\epsilon$, $\alpha=1$, $\alpha-\beta=1/2 -\epsilon$, there exist instances with $n^{\Omega(\log n)}$ $(\alpha,\beta)$-clusters. \end{theorem}
%\vspace{-3mm}

\begin{proof}
Consider the graph $G_{n,p}$ with $p=1/2^l$. % (where l  is a constant (determined later)).
 Consider all ${n \choose k}$ sets of size $k=  \frac{2 \log n}{l} (1-\delta)$, where
 $\delta$ is a constant (determined later).  For each such set $S$,
 the probability it is a clique is $$p^{k \choose 2} \geq (1/2)^{\ell
 k^2/2} = (1/2)^{2 \log^2 n (1-\delta)^2/\ell} = n^{-k(1-\delta)}.$$
 We now want to show that conditioned on $S$ being a clique, it is
 also an $(\alpha,\beta)$-cluster with probability at least $1/2$.  This will imply
 that the {\em expected} number of $(\alpha,\beta)$-clusters is at least $$0.5 {n
 \choose k} n^{-k(1-\delta)} = n^{\Omega(\log n)}.$$
Fix such set of size $k=  \frac{2 \log n}{l} (1-\delta)$.
The probability that a node outside is connected to more than a
$(1/2+\epsilon)$-fraction of the set is upper bounded by
$$2^k \left(\frac{1}{2^l}\right)^{\frac{k}{2}(1+\epsilon)}  \leq n ^{ \frac{2}{l}} 2 ^{-\frac{lk}{2} (1+\epsilon)}= n ^{ \frac{2}{l}}  n^{- (1+\epsilon)(1-\delta)}.$$
By imposing $\frac{2}{l} - (1+\epsilon)(1-\delta) < -1 +\log_n (2)$, we get that this probability is upper bounded by $1/(2n)$.
So by union bound over all nodes we then get the desired result.
We need to impose $(1+\epsilon)(1-\delta) -\frac{2}{l} > 1 + \log_n (2)$. This is true for $\delta \leq \epsilon/4$ and $l > 12/\epsilon$ and $n$ large enough.
\end{proof}

%comment

We note that for  certain range of parameters our bounds in Theorem~\ref{basic-one} this improves over the general upper bound
given in~\cite{mishra-conf,mishra}. Moreover, we show that even in graphs with only one $(\alpha,\beta)$-cluster,
 we show that finding this cluster
is at least as hard solving the {\em planted clique problem} for planted
cliques of size $O(\log n)$, which is believed to
be hard (see, e.g., Hazan and Krauthgamer~\cite{clique-nash}).

\smallskip
\noindent {\bf The Hidden Clique Problem:} In this problem,
the input is a graph on n vertices drawn at random
from the following distribution $G_{n,1/2,k}$ pick a random
graph from $G_{n,1/2}$ and plant in it a clique of size
$k = k(n)$. The goal is to recover the planted clique
(in polynomial time), with probability at least (say)
$1/2$ over the input distribution. The clique
is hidden in the sense that its location is adversarial
and not known to the algorithm. %%(it can be assumed to be planted at random).
%% (but independent of the random graph, e.g. of its degrees).
The hidden clique
problem becomes only easier as $k$ gets larger, and the
best polynomial-time algorithm to date~\cite{aks98}, solves the problem
whenever $k = \Omega(\sqrt{n})$.  Finding a hidden clique for $k=c
\log n$ for any $c$ is believed to be hard. The decision version of this problem is also believed to be hard.%%~\cite{...}.

We begin with a simpler result that finding the {\em
approximately-largest} $(\alpha,\beta)$-cluster is at least as hard as the hidden
clique problem.

\begin{theorem}
Suppose that for $\alpha=1$ and $\beta-\alpha=1/4$, there was an algorithm that for
some constant $c$ could find an $(\alpha,\beta)$-cluster of
size at least
 $MAX/c$, where $MAX$ is size of the largest community with those
 parameters. Then, that algorithm could be used to distinguish (1) a
  random graph $G_{n,1/2}$  from (2) a random graph $G_{n,1/2}$  in which a clique of size
 $2c \log_2(n)$ has been planted.
\end{theorem}
\begin{proof}
We can show that with probability at least $1-1/n$ the largest clique in $G_{n,1/2}$
  largest clique has size $2\log(n)$, which implies the largest
$(\alpha,\beta)$ cluster (with $\alpha=1$ and $\beta-\alpha=1/4$) has size at {\em most} $2 \log(n)$.
On the other hand we can also show that with probability at least $1-1/n$, for $c \geq 8 \ln 2$ the planted clique of size  $2c \log_2(n)$ {\em is} a cluster with
  these parameters.  %[Technically, we might need to assume c \geq 4 or
% so for this last statement, to be able to apply Hoeffding bounds and
% the union bound to say that whp all vertices not in the planted clique
% have edges to at most 3/4 of the vertices inside the planted clique].
  Thus, under the assumption that distinguishing these two cases is
 hard, the problem of finding the approximately-largest $(\alpha,\beta)$-cluster is
 hard.
\end{proof}

We now show that in fact, even finding a single $(\alpha,\beta)$-cluster is as hard
as the hidden clique problem.  Here, instead of $G_{n,1/2}$ we will
use $G_{n,p}$ for constant $p > 1/2$.  Note that the hidden clique
problem remains hard in this setting as well.\footnote{In particular,
if it were easy, then one could solve the decision version of the
hidden clique problem for $G_{n,1/2}$ by first adding additional
random edges and then solving the problem for $G_{n,p}$.  We assume
here that the planted clique has size greater than the largest clique
that would be found in $G(n,p)$.}

\begin{theorem} \label{hardness-one}
For sufficiently small (constant) $\gamma$ and $\epsilon$, with probability at least $1-3/n$, we have that:
(1) the graph
$G_{n,1-\gamma -\epsilon}$ has no $(1,1-\gamma)$ clusters; and (2) a hidden
clique of size $\frac{1}{\epsilon^2} \log n$ is an $(1,1-\gamma)$ cluster.
Therefore, finding even one such cluster is as hard as the hidden clique problem.
\end{theorem}

\begin{proof}
Consider  $G_{n,p}$ for $p=1-\gamma -\epsilon$.
We start by showing that with probability at least $1-1/n$  the size of the largest clique is at most
$\frac{-2 \ln n}{\ln (1-\gamma-\epsilon)}$. For any $k$, the probability that there exists a clique of size $k$ is at most
$${n \choose k} p^{k \choose 2} \leq \frac{n^k}{k!} p^{k^2/2} p^{-k/2}.$$
For $k=\frac{-2 \ln n}{\ln (1-\gamma-\epsilon)}=-2 \log_p {n}$, this is
$$\frac{p^{-k/2}}{k!} n ^{-2 \log_p {n}} p^{2 (\log_p {n})^2}= \frac{p^{-k/2}}{k!} =\frac{n}{k!} = o(\frac{1}{n}).$$

This immediately implies
that with  probability at least $1-1/n$,  $G_{n,p}$ does not contains any $(1,1-\gamma)$ clusters
of size greater than $\frac{-2 \ln n}{\ln (1-\gamma-\epsilon)}$.

We now show that with probability at least $1-1/n$,  $G_{n,p}$ does not contain any $(1,1-\gamma)$ clusters of size
$ \leq \frac{-2 \ln n}{\ln (1-\gamma-\epsilon)}$. For this, we will show that for any set $S$ of size $ \leq \frac{-2 \ln n}{\ln (1-\gamma-\epsilon)}$ and any node $v$ not in $S$, the probability that
$v$ connects to at least $(1-\gamma)|S|$ nodes inside $S$ is at
least $1/\sqrt{n}.$  Because these events are independent over the different nodes $v$, this implies that the probability that no node $v$ outside $S$ connects to at least $(1-\gamma)|S|$ nodes inside $S$ is at most
$\left(1-\frac{1}{\sqrt{n}}\right)^{n-k} \leq e^{-\sqrt{n}/2}$.
By union bound over all sets $S$ of size at most $\frac{-2 \ln n}{\ln (1-\gamma-\epsilon)}$, this will imply that the probability there exits a $(1,1-\gamma)$ cluster of size at most $\frac{-2 \ln n}{\ln (1-\gamma-\epsilon)}$ is at most $1/n$.

Consider a set $S$ of size $k$ and a node $v$ outside $S$.
The probability that $v$ connects to more than $(1-\gamma)k$ nodes inside $S$ is at least
$$ {k \choose \gamma k} (1-\gamma-\epsilon)^{(1-\gamma)k}(\gamma+\epsilon)^{\gamma k} \geq \frac{1}{k} \left(\frac{(1-\gamma) k e}{\gamma k}\right)^{\gamma k}(1-\gamma-\epsilon)^{(1-\gamma)k}(\gamma+\epsilon)^{\gamma k}.$$
This follows from the fact that $${k \choose \gamma k}= \frac{k (k-1) \ldots (k-\gamma k +1) }{(\gamma k)!} \geq \frac{((1-\gamma)k)^{\gamma k}}{k (\gamma k/e)^{\gamma k}}= \frac{1}{k} \left(\frac{(1-\gamma) k e}{\gamma k}\right)^{\gamma k}, $$ where we use the fact that
$(\gamma k)! < 2 \sqrt{2 \pi \gamma k}  (\gamma k/e)^{\gamma k} < k (\gamma k/e)^{\gamma k}$.

So, the probability that $v$ connects to more than $(1-\gamma)k$ nodes inside $S$ is at least
$$\frac{1}{k} (1-\gamma-\epsilon)^{k} \left[ \frac{1-\gamma}{\gamma} \cdot \frac{\gamma+\epsilon}{1-\gamma-\epsilon} e \right]^{\gamma k}\geq [(1-\gamma-\epsilon) e ^{\gamma}]^{k}\frac{1}{k}.$$

This is decreasing with $k$ and thus it suffices to consider $k=\frac{-2 \ln n}{\ln (1-\gamma-\epsilon)}$.
 For this $k$,
 %%using the fact that $x^{1/\ln x}=e$ for all $x$,
 we get  that the probability that $v$ connects to more than $(1-\gamma)k$ nodes inside $S$ is at least
$$\frac{1}{k} e^{-2 \ln n} e ^{- \frac{2 \gamma \ln n}{\ln(1- \gamma -\epsilon)}}= \frac{1}{k} n^{-2 -\frac{2 \gamma}{\ln(1- \gamma -\epsilon)}}.$$
We want this to be greater than $1/\sqrt{n}$, and thus it suffices to have $-2-\frac{2\gamma}{\ln(1-\gamma)}>-0.4$. This holds for $\gamma = 0.1, \epsilon = 0.01.$

Finally, it is easy to show that with probability at least $1-1/n$, a hidden
 clique of size $k=\frac{1}{\epsilon^2}\ln n$  is a  $(1,1-\gamma)$ cluster.
This follows by noticing that every vertex outside the clique has in expectation
$k(1-\gamma-\epsilon)$ connections insides the clique, so by Hoeffding bounds,
the probability it has more than $k(1-\gamma-\epsilon)
+ \epsilon k = k(1-\gamma)$ neighbors inside the clique is at most $1/n^2$.
By union bound, we get that with probability at least $1-1/n$ every vertex outside the clique has at most $k(1-\gamma)$ neighbors inside the clique so the
planted clique is a community as desired.
\end{proof}
\bibliographystyle{plain}
\bibliography{community}

%%%\newpage

\appendix
%\vspace{-3mm}
\section{Additional Proofs}
%\vspace{-1mm}
\subsection{Finding Self-determined Communities in Quasi-Polynomial Time}
\label{additional}
%\section{A  Quasi-Polynomial Algorithm for Enumerating Self-Determined Communities}
We present here a simple quasi-polynomial algorithm for enumerating all the self-determined communities.
\begin{theorem}
\label{basic-one}
For any $\fracconnect$, $\alpha$, $\beta$, $\gamma=\alpha-\beta$, there are  $n^{O(\log{n}/\gamma^2)}$ sets which are  $(\fracconnect,\alpha,\beta)$ (weighted) self-determined communities.
All such communities can be found by using Algorithm~\ref{alg:trivial} with parameters $\fracconnect$, $\alpha$, $\beta$, $\gamma=\alpha-\beta$ and $k(\gamma)= 2 \log {(4 n)} /\gamma^2$.
\end{theorem}
\begin{proof}
Fix a $(\fracconnect,\alpha,\beta)$ (weighted) self-determined community $S$. We show that there exists a multiset $U$ of size $k(\gamma)=2 \log {(4 n)} /\gamma^2$ such that the set $S_U$ of points in $V$ that receive at least
     $(\alpha-\gamma/2)|U|$ amount of vote from points in $U$ is identical to $S$.
The proof follows simply by the probabilistic method.
Let us fix a point $i \in V$.
By Hoeffding, if we draw  a set $U$ of $ 2 \log {(4 n)} /\gamma^2$ uniformly at random from $S$, then with probability $1- 1/(2n)$,
 the average amount of vote that $i$ receives from points in $U$  is within $\gamma/2$ of the average amount of vote that $i$ receives from points in $S$.
 By union bound,  we get that with probability at least $1/2$, for all points in $V$  the average amount of vote that they receive from points in $U$  is within $\gamma/2$ of the average amount of vote that they receive from points in $S$.
 %%he fraction of points in $S$ that vote for them is within $\gamma/2$ of the fraction of nodes in $U$ that vote for them.
 Using this together with the definition of a self-determined community, we get that with probability $1/2$ we obtain $S_U=S$ for $U$ of size  $ 2 \log {(4 n)} /\gamma^2$ drawn uniformly at random from $S$. This then implies that there must exist a multiset $U$ of size $k(\gamma)$ such that $S_U=S$.

Since in Algorithm~\ref{alg:trivial} we exhaustively search over all multisets $U$ (of point from $V$) of size $k(\gamma)$, we
clearly get the list $L$ we output contains all the
 $(\fracconnect,\alpha,\beta)$ (weighted) self-determined communities. Moreover, clearly,  $n^{O(\log{n}/\gamma^2)}$ is an upper bound
on the number of $(\fracconnect,\alpha,\beta)$ (weighted) self-determined communities.
\end{proof}

%\vspace{-2mm}
\begin{algorithm}
  \caption{Algorithm for enumerating self-determined communities}
  {\bf Input:}   Affinity system $(V, \Pi)$, parameters $\fracconnect$, $\alpha$, $\beta$, $\gamma$;  $k(\gamma)$;

\begin{enumerate}
\item[$\bullet$] Set $L=\emptyset$.

\item[$\bullet$] Exhaustively search over all multisets  $U$  with elements from $V$ of size $k(\gamma)$.
\begin{enumerate}
\item[$\bullet$] For $t=1$ to $n$ (determining the meaning of ``vote for'') do: %%
  \begin{enumerate}
  \item[$\bullet$] Let $S_U$ be the subset of points in $V$ that receive at least
     $(\alpha-\gamma/2)|U|$ amount of vote from points in $U$.
      %%I.e., $S_U = \{ i \in V |\vn^{\fracconnect t}_{U}(i) \geq (\alpha-\gamma/2) |U|\}$.
     Add $S_U$ to the list $\listt$.
  \end{enumerate}
  \end{enumerate}
  \item[$\bullet$] Remove from the list $\listt$ all the sets that are not $(\fracconnect,\alpha,\beta)$ weighted self-determined communities.
\end{enumerate}
{\bf Output:}  List of self-determined  communities $L$.
\label{alg:trivial}
\end{algorithm}
%\vspace{-2mm}

\subsection{Additional Proofs in Section~\ref{se:selfdeterminedalgo1}}

 {\sc Theorem~\ref{lb}}
For any constant $\fracconnect \geq 1$, for any $\alpha \geq 2\sqrt{\fracconnect}/n^{1/4}$,
there exists an instance such that the number of $(\fracconnect,\alpha,\beta)$-self-determined communities
 with  $\beta-\alpha=\gamma=\alpha/2$ is $n^{\Omega({1/\alpha})}$.

\smallskip

\begin{proof}
Consider $L=\sqrt{n}$ blobs $B_1$, ..., $B_L$ each of size $\sqrt{n}$. %%Consider $\fracconnect=2$.
Assume that each point ranks the points inside its blob first (in an arbitrary order)
and it then ranks the points outside its blob randomly.
We claim that with non-zero probability for $l \leq n^{1/4}/(2\sqrt{\fracconnect})$ any union of  $l$ blobs  satisfies
the $(\fracconnect,\alpha,\beta)$-self-stability property with parameters $\alpha=1/l$ and $\gamma=\alpha/2$.

Let us fix a set $S$ which is a union of $l$ blobs.
Note that for each point $i$ in $S$, the expected number of points
in $S$ voting for $i$ is $$\sqrt{n}+ (l  \sqrt{n}- \sqrt{n}) \frac{ \fracconnect l \sqrt{n} - \sqrt{n}}{n-\sqrt{n}}.$$ Also, for a point $j$
not in $S$ the expected number of points
in $S$ voting for $j$ is $$l  \sqrt{n} \frac{ \fracconnect l \sqrt{n} - \sqrt{n}}{n-\sqrt{n}} \leq l  \sqrt{n} \frac{\fracconnect l \sqrt{n}}{n}  \leq  \sqrt{n}/4,$$ for $l \leq {n^{1/4}}/{(2\sqrt{\fracconnect})}$.
By Chernoff, we have that the probability that $j$ is voted by more than $\sqrt{n}/2$ is at most $e^{-\sqrt{n}/48}$.

By union bound, we get that the probability that there exists a set $S$ which is a union of $l$ blobs that does not satisfy the $(\fracconnect,\alpha,\beta)$-self-stability property with $\alpha=1/l$, $\gamma=\fracconnect/2$ is at most
 $$n \cdot n^{l/2} \cdot e^{-\sqrt{n}/48}  < 1,$$
for $l \leq {n^{1/4}}/{(2\sqrt{\fracconnect})}$.
 \end{proof}

\medskip

 {\sc Corollary~\ref{number_communities}}
The number of $(\fracconnect,\alpha,\beta)$-self-determined communities  in an affinity system $(V,\Pi)$ satisfies
 $B(n)=n^{ O\left(\log{\left(1/\gamma\right)}/{\alpha} \right)}  \left(\frac{\fracconnect  \log{\left(1/\gamma\right)}}{\alpha}\right)^{O\left(\frac{1}{\gamma^2} \log{\left(\frac{ \fracconnect \log{\left(1/\gamma\right)}}{\alpha \gamma}\right)}\right)}$ and with probability $\geq 1-1/n$ we can find all of them in time $B(n) \poly(n)$.

\smallskip

%%pick delta=1/2.
\begin{proof}
Consider a community size $t$.
For any  $(\fracconnect,\alpha,\beta)$-self-determined community $S$ and let $p_S$ be the probability that $S$ is in the list output by Algorithm in Theorem~\ref{main} with parameters $\fracconnect$, $\alpha$, $\beta$, $t$. By Theorem~\ref{main} we have that $p_S \geq 1-\delta$. By linearly of expectation we have that $\sum_{S}{p_S}$ is the expected number of $(\fracconnect,\alpha,\beta)$-self-determined communities   in the list  output by our algorithm. Combining these, we obtain that
$B(n) (1-\delta) \leq \sum_{S}{p_S} \leq  N_1(\delta) N_2(\delta)$ where   $k_1= \log{\left(16/\gamma\right)}/{\alpha}$,
 $k_2 (\delta)= \frac{8}{\gamma^2} \log{\left(\frac{32 \fracconnect k_1}{\gamma \delta}\right)}$, $N_1(\delta)=n^{k_1}$ and
 $N_2(\delta)= O((2 \fracconnect k_1)^{k_2 (\delta)} \log(1/\delta))$.  By setting $\delta=1/2$, we get the desired bound,  $$B(n)=n^{ O\left(\log{\left(1/\gamma\right)}/{\alpha} \right)}  \left(\frac{\fracconnect  \log{\left(1/\gamma\right)}}{\alpha}\right)^{O\left(\frac{1}{\gamma^2} \log{\left(\frac{ \fracconnect \log{\left(1/\gamma\right)}}{\alpha \gamma}\right)}\right)}.$$

Let $N= N_1(1/2) N_2(1/2) n$.
 By running the algorithm  in Theorem~\ref{main} $2 \log [N]$ times we have that for each  $(\fracconnect,\alpha,\beta)$-self-determined community $S$, the probability that $S$ is not output in any of the runs is at most $(1/2)^{2 \log(N)}\leq 1/N^2$. By union bound, with probability at least $1-1/n$, we output all of them.
\end{proof}

\comment{
\medskip

\noindent {\bf Theorem~~\ref{self-stable:weighted}~}
For any $\fracconnect$, $\alpha$, $\beta$, $\gamma=\alpha-\beta$,
the number of weighted $(\fracconnect,\alpha,\beta)$-self-determined communities is
 $B(n)=(n/\gamma)^{ O\left(\log{\left(1/\gamma\right)}/{\alpha} \right)}  \left(\frac{2\fracconnect  \log{\left(1/\gamma\right)}}{\alpha}\right)^{O\left(\frac{1}{\gamma^2} \log{\left(\frac{ \fracconnect \log{\left(1/\gamma\right)}}{\alpha \gamma}\right)}\right)}$
  and we can find them in time $B(n) \poly(n).$

\begin{proof}
%%\vspace{-1mm}
We perform the reduction in Theorem~\ref{reductionweighted} with $\epsilon=\gamma/2$ and use the algorithm in Theorem~\ref{main}
and the bound in Corollary~\ref{number_communities}.
The proof follows from the fact that  the number of vertices in the new instance
has increased by only a factor of $2/\gamma$.
 %% so the number of communities and the running time will still be polynomial.
We also note that each set output on the reduced instance
can then be examined on the original weighted affinity system, and kept iff it satisfies the community definition with
original parameters.
\end{proof}
}

\subsection{Self-determined Communities in Multi-faceted Affinity Systems}
\label{appendix:facet}
Recall that a multi-faceted affinity system is a system where each node may have more than one
rankings of other nodes. This may reflect, for example, that a person may have two rankings of
other people, one corresponding to personal friends (in descending order of affinity), and one
of co-workers. Suppose that each element $i$ is allowed to have at most $f$ different rankings
$\pi_i^1,\ldots,\pi_i^f$. We say that the pair $(S,\psi)$ is a multi-faceted community where $\psi:S\rightarrow \{1,\ldots,f\}$,
if $S$ is a community where $\psi(i)$ specifies which ranking facet should be used by element $i$.
In other words, as before, let $\phi_{S,\psi}^\ta(i):=|\{s\in S|i\in \pi_s^{\psi(s)}(1:\lceil \ta |S|\rceil)\}|$. Then
$(S,\psi)$ is an $(\al,\be,\ta)$-multifaceted community if for all $i\in S$, $\phi_{S,\psi}^\ta(i)\ge\al |S|$, and for
all $j\notin S$, $\phi_{S,\psi}^\ta(j)<\be |S|$.

For a bounded $f$, it is not harder to find multifaceted communities than to find regular communities. Note that in all our sampling
algorithms can be adapted as follows. Once a representative sample $\{i_1,\ldots,i_k\}$ of the  community $S$ is obtained,
we can guess the facets $\psi(i_1),\ldots,\psi(i_k)$ while adding a multiplicative $f^k$ factor to the running time. We can thus
get the set $S_2$ approximating $S$ in the same way as it is found in Algorithms \ref{alg:selfstablerough} and \ref{alg:selfstablepurify}
while adding a multiplicative factor of $f^{k_1+k_2}$ to the running time. We thus obtain a list $\listt$ that for each multi-faceted community $(S,\psi)$
contains set $S_2$ such that $\Delta(S_2,S)<\ga t/8$:

\begin{claim}
\label{cl:multf}
We can output a list $\listt$ of $(f\cdot n)^{ O\left(\log{\left(1/\gamma\right)}/{\alpha} \right)}  \left(\frac{f\cdot \fracconnect  \log{\left(1/\gamma\right)}}{\alpha}\right)^{O\left(\frac{1}{\gamma^2} \log{\left(\frac{ \fracconnect \log{\left(1/\gamma\right)}}{\alpha \gamma}\right)}\right)}$ sets, such that for each multi-faceted community $S$ there is an $S_2\in \listt$ such that
$\Delta(S_2,S)<\ga t/8$.
\end{claim}

It remains to show that:

\begin{lemma}
\label{lem:multf} Suppose that $(S,\psi)$ is a valid $(\al,\be,\ta)$-multifaceted community of size $t$.
Given $t$ and a set $S_2$ such that $\Delta(S_2,S)<\ga t/8$, there is an algorithm that outputs $S$ with probability $>f^{-8 \log n / \ga^2}/2$.

Moreover,
a facet structure $\psi'$ can be recovered on $S$ so that $(S,\psi')$ is an $(\al-\ga/4,\be+\ga/4,\ta)$-multifaceted community.
\end{lemma}

Together with Claim~\ref{cl:multf}, Lemma~\ref{lem:multf} shows that multifaceted communities can indeed be recovered in polynomial time.
\smallskip

 {\sc Theorem~\ref{thm:mutlf}}
Let $S$ be an $f$-faceted $(\al,\be,\ta)$-community. Then there is an algorithm that runs in $O(n^2)$ time and outputs $S$, as well as a facet structure $\psi'$ on $S$ such that
$(S,\psi')$ is an $(\al-\ga/4,\be+\ga/4,\ta)$-multifaceted community with probability at least
$$
(f\cdot n)^{- O\left(\log{\left(1/\gamma\right)}/{\alpha} \right)}  \left(\frac{f\cdot \fracconnect  \log{\left(1/\gamma\right)}}{\alpha}\right)^{-O\left(\frac{1}{\gamma^2} \log{\left(\frac{ \fracconnect \log{\left(1/\gamma\right)}}{\alpha \gamma}\right)}\right)} f^{-O(\log n/\ga^2)}.
$$

\begin{proof}(of Lemma~\ref{lem:multf}).
The algorithm is very simple. Guess a set $U_2$ of $m=8\log n /\ga^2$ points in $S_2$; guess a function $\psi_2$ on $U_2$; output
$S=$ the set of points that receive at least $(\al-\ga/2)t$ votes according to $(U_2,\psi_2)$.

Note that in the non-faceted case, by Hoeffding's inequality, with probability $>1/2$ selecting a set $U_2$ as above and then selecting
those points that receive at least $(\al-\ga/2)t$ votes from $U_2$ would have yielded $S$. This is because each element of $S$ receives
at least $(\al-\ga/8)t$ votes from elements of $S_2$, while each element of the complement $S^c$ receives at most $(\be+\ga/8)t$ votes
from elements of $S_2$. This reasoning extends to the multifaceted setting, {\em provided}, the function $\psi_2$ coincides with
the function $\psi$ on the elements of $U_2\cap S$. This indeed happens with probability $\ge f^{-|U_2|}=f^{-8 \log n / \ga^2}$, completing the
proof of the first part of the lemma.

For the second part of the lemma we assume that the set $S$ is known and we need to recover the facets $\psi'$ that make $S$ a community. Note
that this step is necessary in order to verify that $S$ is indeed a multifaceted community. There are two cases to consider.

{\bf Case 1:} $t\le 8 \log n / \ga^2$. In this case we can find $\psi$ by exhaustively checking all possibilities in time $O(q^{  8 \log n / \ga^2})$, which is the same
as the probability of success of the first step.

{\bf Case 2:} $t> 8 \log n / \ga^2$. In this case we use linear programming to find a fractional version $\psi_f$ of the function $\psi$ first. In other words, we find
a function $\psi_f:S\times \{1,\ldots,q\}\rightarrow [0,1]$ such that $(S,\psi_f)$ is a ``community" on average:
\begin{enumerate}
\item
for all $s\in S$, $\sum_{i=1}^f \psi_f(s,i)=1$;
\item
for all $x\in S$, $\sum_{s\in S} \sum_{i=1}^f \psi_f(s,i)\cdot \chi_{x\in \pi_s^{i} (1:\ta t)} \ge \al t$;
\item
for all $y\notin S$, $\sum_{s\in S} \sum_{i=1}^f \psi_f(s,i)\cdot \chi_{y\in \pi_s^{i} (1:\ta t)} <\be t$;
\end{enumerate}
This linear program is feasible, since the original $\psi$ is an integral solution to it. As a result, we obtain a fractional solution $\psi_f$ satisfying the three
conditions. To obtain $\psi'$ we round $\psi_f$ by sampling. In other words, we set $\psi'(s)=i$ with probability $\psi_f(s,i)$. By Hoeffding's inequality, since
$t> 8 \log n / \ga^2$, the sampling will preserve conditions 2 and 3 that were imposed on $\psi_f$ up to an additive error of $\ga/4$. Thus, by definition, $(S,\psi')$ will
be an $(\al-\ga/4,\be+\ga/4,\ta)$-multifaceted community.
\end{proof}

%%comment

\subsection {Extensions to weighted affinity systems and to the local model}
\label{addsec:local}

We note that that Algorithm~\ref{alg:selfstableroughlocal} can be combined with our reduction from weighted to unweighted communities to obtain a local algorithm
for finding communities in the weighted case.

Extending the local approach to the multi-faceted setting is more involved, since the definition of $R(v)$ would need to be adapted to this setting. Indeed,
the multi-faceted version $R_f(v)$ of $R(v)$ can be taken to be a random element voted by a random facet $i$ of $v$. Then Algorithm~\ref{alg:selfstableroughlocal}
can be adapted by taking the threshold to be $\frac{(\al-1/2)/(\ta^2 f^2)}{t}$, where $f$ is the number of facets. Note that while an approximation to
any community $S$ can be found locally in near-linear time, finding the exact community $S$ as well as the facet structure on $S$ as in Lemma~\ref{lem:multf} will still
take $f^{O(\log n/\ga^2)}$ time.

\end{document}